\documentclass[11pt]{article}

\usepackage{amssymb,amsthm,amsmath}
\usepackage{libertine}
\usepackage[libertine]{newtxmath}


\usepackage{subfig}
\usepackage[usenames,dvipsnames]{xcolor}
\usepackage[colorlinks,citecolor=blue,linkcolor=BrickRed]{hyperref}\usepackage[colorlinks,citecolor=blue,linkcolor=BrickRed]{hyperref}
\usepackage{fullpage}
\usepackage{graphicx}
\usepackage{xspace}
\usepackage{enumitem}
\usepackage{thmtools,thm-restate}
\usepackage{wrapfig}

\graphicspath{{figures/}}

\newtheorem{theorem}{Theorem}
\newtheorem{cor}[theorem]{Corollary}

\newtheorem{lemma}[theorem]{Lemma}
\newtheorem{obs}[theorem]{Observation}
\newtheorem{defn}[theorem]{Definition}

\newtheorem{remark}[theorem]{Remark}

\newtheorem{fact}[theorem]{Fact}
\newif\ifFULL
\FULLtrue

\newcommand{\IGNORE}[1]{}
\newcounter{note}[section]

\newcommand{\haotian}[1]{\refstepcounter{note}$\ll${\bf Haotian~\thenote:}
  {\sf \color{blue}  #1}$\gg$\marginpar{\tiny\bf HJ~\thenote}}

\global\long\def\cirt#1{\raisebox{.5pt}{\textcircled{\raisebox{-.9pt}{#1}}}}

\newcommand{\pbintdisc}{Interval Discrepancy\xspace}
\newcommand{\pbone}{Online Interval Discrepancy\xspace}
\newcommand{\pbtreebalance}{Online Tree Balancing\xspace}
\newcommand{\pbtwo}{Online Stripe Discrepancy\xspace}
\newcommand{\pbthree}{Online Envy Minimization\xspace}

\usepackage{tikz}
\usetikzlibrary{arrows}
\usetikzlibrary{arrows.meta}
\usetikzlibrary{shapes}
\usetikzlibrary{backgrounds}
\usetikzlibrary{positioning}
\usetikzlibrary{decorations.markings}
\usetikzlibrary{patterns}
\usetikzlibrary{calc}
\usetikzlibrary{fit}
\usetikzlibrary{snakes}
\tikzset{
    >=stealth',
    pil/.style={
           ->,
           thick,
           shorten <=2pt,
           shorten >=2pt,}
}
\tikzset{->-/.style={decoration={
  markings,
  mark=at position .5 with {\arrow{>}}},postaction={decorate}}}

\newcommand{\calC}{\mathcal{C}}

\newcommand{\A}{\ensuremath{\mathcal{A}}\xspace}
\newcommand{\B}{\ensuremath{\mathcal{B}}\xspace}
\newcommand{\p}{\ensuremath{\mathbb{P}}}

\newcommand{\E}{\mathbb{E}}

\newcommand{\T}{\mathcal{T}}

\newcommand{\eat}[1]{}
\newcommand{\hide}[1]{{\Large \color{red} Contents here are hidden! To reveal contents, remove this command.}}

\newcommand{\calP}{\ensuremath{\mathcal{P}}}

\newcommand{\envy}{\ensuremath{\mathsf{envy}}\xspace}

\newcommand{\child}{\ensuremath{\mathsf{Child}}}
\newcommand{\disc}{\ensuremath{\mathsf{disc}}}

\newcommand{\polylog}{\ensuremath{\mathsf{polylog}}\xspace}
\newcommand{\poly}{\ensuremath{\mathsf{poly}}}

\newcommand{\col}{\ensuremath{\chi}}

\newcommand{\dang}{\mathsf{dangerous}}

\newcommand{\calS}{\mathcal{S}}
\newcommand{\valuation}{\mathbf{v}}

\newcommand{\Def}{\textrm{def}}

\allowdisplaybreaks

\newenvironment{proofof}[1]{\smallskip\noindent{\bf Proof of #1.}}%
        {\hspace*{\fill}$\Box$\par}

\title{ 
Online Geometric Discrepancy for Stochastic Arrivals\\ with Applications to Envy Minimization
}

\author{
Haotian Jiang
\thanks{Paul G. Allen School of Computer Science \& Engineering, University of Washington, Seattle, USA. Email: \texttt{jhtdavid@cs.washington.edu.}}
\and Janardhan Kulkarni
\thanks{Microsoft Research, Redmond, USA. Email:\texttt{jakul@microsoft.com}.}
\and Sahil Singla
\thanks{Computer Science Department at Princeton University and School of Mathematics at Institute for Advanced Study, USA. Email:\texttt{singla@cs.princeton.edu}.}
}

\date{\today }

\begin{document}
\maketitle
\thispagestyle{empty}

\begin{abstract}

Consider a unit interval $[0,1]$ in which $n$ points arrive one-by-one independently and uniformly at random.  On arrival of  a point, the problem is to immediately and irrevocably color it in $\{+1,-1\}$ while ensuring that every interval $[a,b] \subseteq [0,1]$ is nearly-balanced. We define \emph{discrepancy} as the largest imbalance of any interval during the entire process. If all the arriving points were known upfront then we can color them alternately to achieve a discrepancy of $1$. What is the minimum possible expected discrepancy when we color the points online? 

We show that the discrepancy of the above problem is sub-polynomial in $n$  and that no algorithm can achieve a constant discrepancy. This is a substantial improvement over the trivial random coloring that only gets an $\widetilde{O}(\sqrt n)$ discrepancy.  We then obtain similar results for a natural generalization of this problem to $2$-dimensions where the points arrive uniformly at random in a unit square. This generalization allows us to improve recent results of Benade et al.~\cite{BenadeKPP-EC18} for the \emph{online envy minimization} problem when the arrivals are stochastic.

\end{abstract}
\newpage


\setcounter{page}{1}


\section{Introduction} \label{sec:intro}

Given a set $V$ of $n$ elements and a set system $\calS \subseteq 2^{V}$, the (combinatorial) discrepancy minimization problem is to color  the elements $\col \in \{+1,-1\}^{V}$  to minimize the maximum \emph{imbalance} of a set $S \in \calS$, i.e., we want to find  \emph{discrepancy} of set system $\calS$:
\[
{ \disc(\calS)  \overset{\Def}{=} \min_{\col} \max_{S \in \calS} \Big|\sum_{i\in S} \col(i) \Big|.}
\]
\IGNORE{
\haotian{The term discrepancy is defined with respect to a coloring?
\[
{ \disc(\calS, \col)  \overset{\Def}{=} \max_{S \in \calS} \Big|\sum_{i\in S} \col(i) \Big|.}
\]}}

This problem has been extensively studied because of its various applications in  approximation algorithms, pesudorandomness,  irregularities of distributions, sparsification, and differential privacy; see~\cite{Matousek-Book09,Chazelle-Book01, HobergR17, Nikolov-Thesis14,Bansal-Notes19} for more details.  
For general set systems, it is easy to show that a random assignment of colors gets $O(\sqrt{n \log |\calS|})$ discrepancy. 
Much of the discrepancy theory deals with when can this trivial coloring  be improved.
In a seminal result, Spencer showed that for any $\calS$ one can beat random coloring to obtain $O\big(\sqrt{n\log (|\calS|/n)}\big)$ discrepancy~\cite{Spencer85}.
Another important  line of work bounds the discrepancy in terms of frequency parameter $t$, which is the maximum number of sets in $\calS$ in which an element appears. In particular,  the result of Beck and Fiala ~\cite{BeckFiala-DAM81} says that the  discrepancy is at most $2t-1$ and   Banaszczyk's bound~\cite{Banaszczyk-Journal98} gives the  discrepancy of  $O(\sqrt{t \log n})$. 
A prominent open question in the field asks whether one can get $O(\sqrt t)$ discrepancy in the Beck-Fiala setting, which would generalize Spencer's result to sparse set systems.
On the algorithmic front, since the breakthrough result of Bansal \cite{Bansal-FOCS10}, there has been a remarkable progress in getting polynomial time algorithms matching these bounds~\cite{Lovett-Meka-SICOMP15, Bansal2013, BansalDGL18, BansalG17, BansalDG16, LevyRR17, Rothvoss14, EldanS18}). 

When the set systems have additional structures, much smaller  bounds on the discrepancy are often possible compared to the two bounds mentioned above.
Geometric set systems are some of the well known examples.
The simplest case is that  we are given $n$ points on the unit interval $[0,1]$ and $\calS$ is formed by all sub-intervals. 
Here  $\disc(\calS) \leq 1$ because we can color  odd and even points alternatively $+1$  and $-1$. 
A more interesting example is the classic  Tus\'nady's problem, where we are given $n$ points in a unit square and  $\calS$ consists of all  axis-parallel rectangles. Here the discrepancy is known to be between $\Omega(\log n)$~\cite{Beck-Combinatorica81} and $O(\log n)^{1.5}$~\cite{Nikolov-Mathematika19}.  
A line of work~\cite{EzraL19, BansalMeka-SODA19,HobergRothvoss-SODA19, Franks19} also bounds the discrepancy of {\em stochastic} set systems in the Beck-Fiala setting, where 
one can indeed obtain $O(\sqrt t)$ discrepancy (under some very mild assumptions)~\cite{BansalMeka-SODA19, EzraL19}.

\IGNORE{If the number of elements is much larger than the number of sets, this bound can be further improved to $O(1)$ \cite{EzraL19, HobergRothvoss-SODA19,Franks19}.}


Can we design \emph{online} algorithms for the discrepancy minimization problems that beat random coloring? This question was first posed by Spencer~\cite{Spencer77}.
In the online setting, the elements arrive one-by-one and 
upon arrival of an element we know the sets to which the element belongs.
An online algorithm has to immediately and irrevocably color the elements without knowing the future input.
Very recently, Bansal and Spencer~\cite{BansalSpencer-arXiv19} study this question in the context of stochastic online vector balancing problem. 
In the online setting, stochasticity is also a {\em necessary} assumption as for adversarial arrivals it is known that no online algorithm can achieve a smaller  discrepancy than random coloring~\cite{Spencer-Book87}  (also see~Chapter~15 in \cite{alonspencer}).
Bansal and Spencer \cite{BansalSpencer-arXiv19} show that for the online vector balancing problem with random inputs, one can get $O(\sqrt n)$ discrepancy matching the offline result.

In this paper, we continue this line of investigation and study  \emph{online}  discrepancy minimization for  \emph{geometric}   set systems. We show in Section~\ref{sec:lowerbounds} that again nothing better than random coloring is possible for adversarial arrivals. This leads us to the following basic question:
\vspace{-0.2cm}
\begin{quote} 
\emph{Can  we design \emph{online} algorithms for stochastic inputs that (approximately) achieve the smaller offline discrepancy bounds of geometric set systems?}
\end{quote}
\vspace{-0.2cm}

We believe that online geometric discrepancy minimization problems are interesting in their own right. However, specific problems we study in this paper are motivated by  applications to online envy minimization, which we discuss in Section~\ref{subsec:OnlineEnvyMin}.


\subsection{Our Results}
The first problem that we consider is the stochastic analog of  $n$ points on the unit interval. 
%

\noindent \textbf{\pbone:} \emph{Suppose $n$ points arrive one-by-one independently and uniformly at random on the unit interval $[0,1]$. The set system $\calS$ consists of all \emph{intervals} $[a,b]$ for $0\leq a< b\leq 1$. If we have to 
immediately and irrevocably color an element on its arrival, what is the minimum possible expected discrepancy?}

%

As mentioned above, this problem is trivial in the {offline setting}. We get discrepancy $1$ by alternately coloring points $\{+1,-1\}$. For online decisions the answer is no longer straightforward since we do not know if the next element will be odd or even in the final order. Indeed, in Section~\ref{sec:lowerbounds} we show that no online algorithm can obtain a constant discrepancy. Randomly coloring the elements $\{+1,-1\}$ gives only an $\widetilde{O}(\sqrt{n})$ discrepancy. Can we  beat random coloring? We answer this affirmatively.

\begin{restatable}{theorem}{thmpbone}
\label{thm:pbone}
There is an efficient online algorithm for the \pbone problem that gives $O(n^{c/\log\log n})$  discrepancy w.h.p. for some universal constant $c$. 
\end{restatable}

Throughout the paper, ``w.h.p.'' stands for ``with high probability'' and it means with $1-1/\poly(n)$ probability where the exponent of the polynomial can be made as large as desired, depending on the constant $c$. 
The assumption of stochastic arrivals is crucial in obtaining $o(\sqrt{n})$ discrepancy. In Section~\ref{sec:lowerbounds}, we show discrepancy is $\Omega(\sqrt{n})$ for any online algorithm with adversarial arrivals. 

Next we generalize \pbone to two-dimensions. Here, points arrive uniformly at random in a unit square and the goal is to minimize  \pbone after projecting the points on both the axes.

\noindent \textbf{\pbtwo:} \emph{Suppose $n$ points arrive one-by-one independently and uniformly at random\footnote{In fact, our approach can be used to handle any product distribution on the unit square. One can use the probability integral transformation to reduce any product distribution to the uniform distribution without increasing the discrepancy.} on the unit square $[0,1]\times [0,1]$. The set system $\calS$ consists of all \emph{stripes} $[a,b]\times [0,1]$ and $[0,1]\times [a,b]$ for $0\leq a <b \leq 1$.  If we have to immediately and irrevocably color an element on its arrival, what is the minimum possible expected discrepancy?}


The above problem has  also been studied in the offline setting where we know the location of the $n$ points upfront. Usually it is stated as given two permutations on $n$ elements, color the elements to minimize the discrepancy of every interval of both the permutations. The  problems are equivalent as the two permutations correspond to the orders given after projecting the points on both the axes. 
A clever proof of Spencer~\cite{Spencer-Book87} shows that in the offline setting the discrepancy is always bounded by $2$.\footnote{Whether the discrepancy  for three permutations is $O(1)$ was a ``tantalizing'' open question~\cite{Matousek-Book09}. It was finally resolved by Newman et al.~\cite{NewmanNN-FOCS12} who showed that the discrepancy can be $\Omega(\log n)$.} In the online setting, again randomly coloring the elements $\{+1,-1\}$ gives an $\widetilde{O}(\sqrt{n})$ discrepancy, and the question is if we can obtain smaller upper bounds.

\begin{restatable}{theorem}{thmpbtwo}
 \label{thm:StripeProblem}
There is an efficient online algorithm for the \pbtwo problem that gives $O(n^{c/\log\log n})$  discrepancy w.h.p. for some universal constant $c$. 
\end{restatable}

Besides being a natural problem, next we show that the \pbtwo problem has applications to  envy minimization.


\subsection{Applications to  Envy Minimization}
\label{subsec:OnlineEnvyMin}

The goal of fair division is to allocate items between competing players ``fairly''. A popular measure of fairness is an \emph{envy-free} solution, i.e., everyone values their allocation more than any other player's allocation~\cite{FoleyEssay67,ThomsonVarian-Essay85}. When the items are \emph{indivisible}, however, envy-free allocations are not always possible, e.g., a single item and two players. So instead we want an allocation that minimizes envy.

Given a set $V$ of $n$ indivisible items  and valuations $\valuation^i = (v^i_1, v^i_2, \ldots, v^i_n)$ of two players for $i\in \{1,2\}$, the \emph{envy minimization} problem is to {allocate} these items to the players, i.e. find a subset $S \subseteq V$ for the first player and the remaining items $\overline{S} = V\setminus S$ for the second player, to minimize
\begin{align} 	\label{eq:envyDefn}
 \envy(\valuation^1, \valuation^2 , S , \overline{S}) \overset{\Def}{=}	\max \Big\{ v^1(\overline{S}) - v^1(S)~,~ v^2(S) - v^2(\overline{S}) \Big\}
\end{align}
where $v^i(S)$ denotes $\sum_{j\in S} v^i_{j}$. A simple round-robin algorithm where the players alternately select their best of the remaining items ensures    envy is at most $\max_{i,j} \{v^i_j\}$~\cite{LiptonMMS-EC04,Budish-Journal11}.

Motivated by applications in food banks, a recent work of Benade et al.~\cite{BenadeKPP-EC18} studies an \emph{online} version of envy minimization. Here items arrive one-by-one, i.e. on arrival item $j$ reveals its valuations $v^i_j$, and the algorithm has to immediately and irrevocably  allocate the item. Assuming that all the valuations lie in $[0,1]$, \cite{BenadeKPP-EC18} show that the minimum possible envy is  $\widetilde{\Theta}(\sqrt{n})$ and is achieved  by the trivial algorithm that randomly allocates the items. Since this bound is tight for adversarial arrival of items, and also because the algorithm is uninteresting,  we ask  whether smaller envy is possible for stochastic arrivals.

\noindent \textbf{\pbthree:} \emph{Given probability distributions $D_i$ over $[0,1]$ for $i\in \{1,2\}$, suppose $n$ items independently  draw   their valuations $v^i_j \sim D_i$  for $j \in [n]$\footnote{By $[n]$ we denote the set $\{1,2,\ldots,n\}$.}.  If the  valuations $v^i_j$ of these items are  revealed one-by-one and we have to  immediately and irrevocably allocate an item when its valuations are revealed, what is the minimum possible expected envy?}

Our next result is to reduce the \pbthree problem to the \pbtwo problem.  This allows us to obtain the following result using Theorem~\ref{thm:StripeProblem}.

\begin{restatable}{theorem}{thmpbthree}
 \label{thm:envyMin}
There is an efficient online algorithm for the \pbthree  problem that gives $O(n^{c/\log\log n})$  envy w.h.p. for some universal constant $c$. 
\end{restatable}

The proof of Theorem~\ref{thm:envyMin} goes via a stronger notion of envy which might be of independent interest. We show that our bound on   \pbtwo  in Theorem~\ref{thm:StripeProblem} implies a bound  on the ``ordinal envy'', which in turn implies a bound on the ``cardinal envy'' defined in~\eqref{eq:envyDefn}. Here, the ordinal envy of a player is essentially the worst cardinal envy that is consistent with a particular ordering (Lemma~\ref{lem:OEnvyEqualsWorstVEnvy}). 

\IGNORE{Here,  ordinal envy of a player roughly denotes the  number of items in another player's set that cannot be matched to a higher value item in their set (Definition~\ref{defn:OrdinalEnvy}).}



\subsection{Our Approach via \pbtreebalance}

The  approach of all our results is to go via  the following  \emph{\pbtreebalance} problem. We think that this problem is of independent interest and will find further applications. 

\noindent \textbf{\pbtreebalance:} \emph{Given a complete $m$-ary tree of height $h$, suppose $n$ points arrive one-by-one independently and uniformly at random at the leaves of this tree with possible repetitions. The set system $\calS$ consists of all subtrees, i.e., all arrivals in a subtree correspond to a set in $\calS$.  If we have to immediately and irrevocably color an element on its arrival, what is the minimum possible expected discrepancy?
}

\IGNORE{
\begin{wrapfigure}{L}{0.45\textwidth}
\vspace{-7mm}
	
\input{figs/tree}

\vspace{-10mm}


\label{fig:tree}
\end{wrapfigure}}

The idea for defining  \pbtreebalance  is that we can ``approximately'' reduce \pbone to it by embedding  the unit interval onto a tree. This is achieved by breaking the unit interval into $m^h$ disjoint pieces, and an arrival in the $i$th piece $[\frac{i-1}{m^h},\frac{i}{m^h} ]$ corresponds to an arrival in the $i$th leaf. The losses due to subintervals within any piece are small and can be easily bounded.


\smallskip
\noindent \textbf{Solving \pbtreebalance.} 
Our algorithm is based on a  potential function $\Phi(\cdot)$. More precisely, after $t$ arrivals we define $\Phi(t)$ as the sum of hyperbolic cosines (recall, $\cosh( x) = \exp( x)/2+ \exp(- x)/2$) of the \emph{imbalance} of every subtree (see~\eqref{eq:potential}).  The algorithm simply greedily assigns the next arrival a color such that the increase in the potential  is minimized. 
The use of hyperbolic cosine 
is natural here because it behaves like the softmax function but does not depend on the sign of  imbalance, e.g., see Chazelle's book~\cite{Chazelle-Book01} and Bansal-Spencer's recent paper~\cite{BansalSpencer-arXiv19}. 
We show that w.h.p. the potential of our algorithm is always bounded by $\poly(n)$, which directly implies that the imbalance of every subtree (and hence  discrepancy) is ``small''. To achieve this, our main claim is that the potential has  a ``drift'' towards $0$. Roughly, we prove that if $\Phi(t) > n^{10}$ then $\E[\Delta \Phi(t) ]<0$, so that  w.h.p. the potential  remains always smaller than $n^{20}$ (Lemma~\ref{lem:RandTreeBal}).  Most of our effort goes in proving the existence of this drift. 
 
To analyze the drift, we use the standard idea of  bounding $\Delta \Phi(t) $  using the Taylor expansion  of $\Phi(t)$. It is not difficult to see that this gives an expression of the form
 \[\Delta \Phi(t) ~\leq~  L\cdot  \col(t) + Q,
 \]
 where $L = \sum_{i \in [h]} \sinh(d_i)$ is the sum of  hyperbolic sines of a vector of imbalances  
  $\mathbf{d} = (d_1, \ldots, d_h)$ along a root-leaf path in the tree and  $Q = \sum_{i \in [h]} \cosh(d_i)$  is the sum of hyperbolic cosines of the \emph{same} imbalance  vector $\mathbf{d}$. Since we are free to choose $\col(t) \in \{-1,+1\}$, we get 
$ \Delta \Phi(t) \leq  -  |L|  +  Q$ (notice $Q$ is always non-negative). Thus, to prove $\Delta \Phi(t) <0$, it suffices to show $|L| \approx Q$. 
Alternately, since $|\sinh(x)| = \sqrt{\cosh^2(x)-1} \approx \cosh(x)$, it suffices to show that for $\mathbf{d}$  the magnitude function roughly ``separates'' over the sum of hyperbolic sines, i.e.,
\begin{align}\label{eq:sepLemStat}
\textstyle{\Big| \sum_{i\in [h]} \sinh(d_i) \Big| \approx \sum_{i\in [h]} |\sinh(d_i)|.}
\end{align}

\noindent \textbf{A Separation Lemma on Trees.}
A separation statement like~\eqref{eq:sepLemStat} is clearly false for an arbitrary imbalance vector $\mathbf{d}$, e.g., consider $\mathbf{d} = (+1,-1,+1,-1,\ldots)$ where the  LHS is (close to) zero. The heart of our proof lies in proving  a \emph{Separation Lemma} (see Lemma~\ref{lem:ScalarSeparationLemma}) that for a uniformly random root-leaf path in the tree,  in expectation the randomly generated imbalance vector $\mathbf{d}$ will satisfy~\eqref{eq:sepLemStat}. 
Since the arrivals are uniformly random, we can exploit the tree structure of our problem and use induction on the height of the tree. The key to our inductive proof is the definition of a ``safe subtree'' which is a subtree rooted at a child $s \in \child(r)$ of the root $r$ such that $|\sinh(d_s) + \sinh(d_r)|$ roughly separates into $|\sinh(d_s)| + |\sinh(d_r)|$. Any subtree that doesn't satisfy this property is called a ``dangerous subtree''. 
We crucially use the fact that the imbalance of the root $r$ equals the sum of the imbalances of its children, which is true because the children partition arrivals in their parent. 
Next we show that this implies that a large fraction of the subtrees of the root $r$ are safe, which allows us to apply induction hypothesis directly to the safe subtrees in the case where they are ``heavy'', i.e., they constitute a large fraction of the mass of  $Q$. On the other hand, if the safe subtrees are ``light'', we cannot directly apply induction hypothesis to the dangerous subtrees as they might cancel out the value of $\sinh(d_r)$ from the root. Nevertheless, we show that there exists a way to modify the imbalance of the  subtrees by incurring small losses so that we can apply induction hypothesis to the dangerous subtrees after the modification.


A lemma similar in spirit to our Separation lemma was also shown in the recent work of Bansal and Spencer for online vector balancing~\cite{BansalSpencer-arXiv19}. Their proof, however,  quite crucially exploits the fact that each coordinate is uniformly and independently distributed in $\{+1,-1\}$
\footnote{It is unclear how to extend their proof to the case where each coordinate is uniformly and independently in $\{0,1\}$.}. 
In contrast, in our setting each coordinate is in $\{0,1 \}$ and there are {\em correlations} among the coordinates due to the tree structure.
These correlations along with the fact that we want to show a discrepancy bound  significantly smaller than $O(\sqrt n)$ introduce several new technical difficulties in our problem. 
But on the other hand, the tree structure allows us to establish the necessary  properties required for our inductive proof.

\smallskip
\noindent \textbf{Solving \pbtwo.} 
We again go via \pbtreebalance. This time we embed the unit square $[0,1]\times [0,1]$ into two trees, one for each dimension after projecting the points on the corresponding axis. The new potential function $\Phi(t)$ equals $\Phi_x(t) + \Phi_y(t)$ where $\Phi_x$ and $\Phi_y$ are defined for the trees corresponding to both the axes. We again argue that this potential is  bounded by $\poly(n)$ due to a drift towards $0$. The primary difference is that now we have
\[ \Delta \Phi(t) ~\leq~  \big(L_x\cdot  \col(t) + Q_x \big) + \big(L_y \cdot  \col(t) + Q_y \big).
\]
So although  the $1$-d argument implies there is a color with $\E\big[-|L_x| + Q_x\big] < 0$ and a color with $\E \big[-|L_y| + Q_y\big] < 0$,  it is not clear if the two colors are consistent with each other. In other words, the challenge is that the 
the two axes might cancel the effect of each other. 
We overcome this  by using the independence of the $x$ and $y$ coordinates to argue that with $\Omega(1)$ probability such a cancellation does not happen.



\subsection{Further Related Work and Open Problems}

The \emph{$\ell$-permutations} problem consists of $\ell$ permutations over $n$ elements, and the goal is to minimize discrepancy over every interval of  the permutations. As mentioned earlier, for $\ell \in \{1,2\}$ we know  discrepancy is $O(1)$~\cite{Spencer-Book87}, but for $\ell=3$ discrepancy becomes $\Theta(\log n)$~\cite{NewmanNN-FOCS12}. Our results in Theorem~\ref{thm:pbone} and Theorem~\ref{thm:StripeProblem} can be viewed as obtaining sub-polynomial bounds for the $1$-permutation and $2$-permutation problems in an online model with stochastic arrivals. An immediate open question is whether one can obtain $\polylog~n$ bounds? 
In Section~\ref{sec:lowerbounds} we  give for these problems $\Omega(\sqrt{n})$ lower bounds  for adversarial arrivals and $\Omega(\polylog n)$ lower bounds for stochastic arrivals. 
For general $\ell$, it is known how to obtain an $O(\sqrt{\ell} \log n)$ discrepancy in the offline setting~\cite{SpencerST-SODA97}.  It will be interesting to extend our online results to $\ell$ permutations (dimensions), i.e., for uniform arrivals in $[0,1]^{\ell}$. 
Another nice question is to get $\polylog~n$ bounds for the online stochastic Tus\'nady's problem where $\calS$ consists of every axis parallel rectangles.

\IGNORE{
The  best known bounds on discrepancy in the Beck-Fiala setting (i.e., each element is in at most $t$ sets) are $2t - \log^* t$~\cite{Bukh-Journal16} and $O(\sqrt{t \log n})$~\cite{Banaszczyk-Journal98,BansalDG-FOCS16}. This problem has been recently also studied for stochastic set systems where $O(\sqrt{t})$ discrepancy is possible~\cite{EzraLovett-Random19,BansalMeka-SODA19,HobergRothvoss-SODA19}.  Our work (along with Bansal-Spencer~\cite{BansalSpencer-arXiv19}) initiates the study of these stochastic systems when the algorithm also makes online decisions. This raises the question if one can obtain a $\poly\big(t\log(n)\big)$ discrepancy using an online algorithm for stochastic set systems.
}

For  various notions of envy free allocations  we refer the readers to~\cite{AzizGMW-AI15}.
The  envy minimization problem in the online setting was first considered  by Benade et al.~\cite{BenadeKPP-EC18}. In their model an adversary adaptively decides the values of the next item for both the players in $[0,1]$. They give a deterministic algorithm with  $\widetilde{O}(\sqrt{n})$ envy, and show that their bound is tight up to $\polylog~n$ factors. Our Theorem~\ref{thm:envyMin} shows that one can obtain much smaller envy under the assumption of stochastic item values. It is an interesting open question if our results can be  extended  to a setting where values of the two players are correlated for the same item (but are i.i.d. over different items). Extending our results to more than two players will also be interesting.


\subsection{Roadmap}
In Section~\ref{sec:OnlineTreeBal} we discuss the \pbtreebalance problem. We prove a Separation Lemma  and  use it to obtain $O(\polylog(n))$ bounds on the discrepancy of trees of height $h = O(\log\log n)$. In Section~\ref{sec:1DTusnady} we use our  \pbtreebalance results to obtain sub-polynomial bounds on the discrepancy for \pbone. Here we prove Theorem~\ref{thm:pbone} and show a super-constant lower bound. We also discuss lower bounds for \pbone in adversarial settings in Section~\ref{sec:lowerbounds}. In Section~\ref{sec:OnlineRand2DStripe} we prove Theorem~\ref{thm:StripeProblem} to obtain sub-polynomial bounds on the discrepancy for \pbtwo. Finally, in Section~\ref{sec:envyMinim} we discuss our applications to envy minimization. In  particular, we discuss the notion of ordinal envy and prove Theorem~\ref{thm:envyMin} by reducing  \pbthree to \pbtwo via the ordinal envy.



\section{\pbtreebalance}
\label{sec:OnlineTreeBal}

In this section we focus on \pbtreebalance and obtain $O(\polylog(n))$ discrepancy for it. 
In Section~\ref{sec:1DTusnady} we show how this   immediately implies a sub-polynomial discrepancy for \pbone.

Suppose we have a complete $m$-ary tree $\T$ with root $r$ and height\footnote{Define the height of the trivial tree that contains a single node to be $0$.} $h$.
The {\em imbalance} $d_v$ of each node $v \in \T$ is initially 0.
There are $n$ online arrivals where  the $t$th arrival picks a root-leaf path $\calP_t$ uniformly at random.
Upon picking $\calP_t$, we need to immediately and irrevocably assign $\col(t) \in \{-1,+1\}$ to arrival $t$, which updates the imbalance of all nodes $v$ as:
\[
d_v ~\leftarrow~ d_v +\col(t) \qquad \forall v \in \calP_t.
\]
The  goal is to minimize the discrepancy of the tree, i.e., the worst imbalance after $n$ arrivals:
\[
\disc(\T)  ~\overset{\Def}{=} ~ \max_{v \in \T} \{ |d_v|\}.
\]
If we randomly color each arrival then the discrepancy can be bounded by $\widetilde{\Theta}(\sqrt{n})$. The following is our main result for \pbtreebalance which  shows one can do much better for height $h = O(\log \log n)$ and fan-out $m = \Omega(1)$. The assumptions in Theorem~\ref{thm:RandTreeBalMainLemma} will be satisfied\footnote{For our purpose of proving Theorem~\ref{thm:pbone}, we will be taking $h = \log \log n/C$ for sufficiently large constant $C$ and $m = n^{\frac{1}{h+1}} \gg 100$.} in our reduction from \pbone in Section~\ref{subsec:TreeBaltoRand1DTus}.

\IGNORE{
\begin{figure}[]
\centering
\input{figs/tree}
\caption{An illustration of a  random trajectory $\calP_t$ on a tree with $m=2$ and $h=4$.} 
\label{fig:tree}
\end{figure}
}

\begin{theorem}
\label{thm:RandTreeBalMainLemma}
For \pbtreebalance with $n$ arrivals on a complete $m$-ary tree of height $h \leq \log \log n/C$ for sufficiently large constant $C$ and  fan-out $m \geq 100$, 
there is an efficient algorithm that satisfies w.h.p.
\[
\disc(\T)  ~=~  O( \log^2(n)).
\]
\end{theorem}

The remaining section focuses on the proof of Theorem~\ref{thm:RandTreeBalMainLemma}.  

\IGNORE{
\begin{theorem}
\label{thm:RandTreeBalMainLemma}
For \pbtreebalance with $n$ arrivals on a complete $m$-ary tree of height $h = O(\log \log n)$,  fan-out $m = \Omega(\log^3 n)$, and the number of nodes $(m^{h+1}-1)/(m-1) < n$, 
there is an efficient algorithm that satisfies w.h.p.
\[
\disc(\T)  =  O( \log^2(n)).
\]
\end{theorem}
}

\subsection{Proof of Theorem~\ref{thm:RandTreeBalMainLemma} using a Potential Based Algorithm} 
\label{subsec:AlgTreeBal}
To describe our algorithm we need a potential function $\Phi(t)$ for $t \in [n]$ and $\lambda \overset{\Def}{=} \frac{1}{\log n}$:
\begin{equation}
\label{eq:potential}
\Phi(t) ~ \overset{\Def}{=} ~ \sum_{v \in \T} \cosh(\lambda d_v).
\end{equation}
The algorithm simply assigns $\col(t) \in \{-1,+1\}$ to minimize the  increase in the potential.

Our plan to proving Theorem~\ref{thm:RandTreeBalMainLemma} is to show that the final potential value $\Phi(n)$ is polynomially bounded. This suffices because it guarantees every $d_v$ is at most $O( \log^2(n))$, as otherwise the potential becomes super-polynomial.


To analyze the increase in  potential at the $t$th arrival, we use the standard step of bounding $\Delta \Phi(t)$ by the Taylor expansion (e.g., see~\cite{BansalSpencer-arXiv19}). Since $\cosh'(x) = \sinh(x)$ and $\sinh'(x) = \cosh(x)$, we get
\begin{align*}
\Delta \Phi(t) & = \sum_{v \in \calP_t}\Big( \lambda \sinh(\lambda d_v) \cdot  \col(t) + \frac{\lambda^2}{2!} \cosh(\lambda d_v) \cdot  \col(t)^2 + \frac{\lambda^3}{3!} \sinh(\lambda d_v) \cdot  \col(t)^3 + \ldots \Big) \\
 &\leq \lambda \Big(\sum_{v \in \calP_t} \sinh(\lambda d_v) \Big)\cdot  \col(t) + \lambda^2 \Big( \sum_{v \in \calP_t} \cosh(\lambda d_v) \Big),
\end{align*}
where we use $|\sinh(x)| \leq \cosh(x)$ for all $x$, magnitude $|\col(t)| = 1$, and that $\lambda=o(1)$. By defining
\[ L \overset{\Def}{=} \sum_{v \in \calP_t} \sinh(\lambda d_v) \quad \text{ and } \quad Q \overset{\Def}{=} \sum_{v \in \calP_t} \cosh(\lambda d_v),
\]
we  can rewrite the last inequality as
\begin{align*}
\Delta \Phi(t) \quad \leq  \quad \lambda L \cdot \col(t) + \lambda^2 Q \qquad \leq \qquad -\lambda |L| + \lambda^2 Q.
\end{align*}
Here, the second inequality is because we can assign $\col(t)\in \{-1,+1\}$ such that $L\cdot \col(t) = -|L|$, and since our algorithm picks $\col(t)$ to minimize $\Delta \Phi(t)$, the  increase in potential of our algorithm could only be smaller. Below we will show that $\Phi(t)$ can never be very large because   otherwise  $\E[\lambda |L|] > \E[\lambda^2 Q]$, so we have $\E[\Delta \Phi(t)] <0$.
To state this formally, we need some  notation. Let $\beta = 100$ be a constant. Define $f(h)$ to be an increasing function of $h$ with $f(0) = 4$ and $f(h) = 200 \beta f(h-1)$ for $h \geq 1$, i.e., $f(h) = 4 \cdot (200 \beta)^h$.

\IGNORE{\begin{itemize}[topsep=0mm,itemsep=0mm]
    \item Let $\beta = 100$ be a constant.
    \item Let $f(h)$ be some increasing function of $h$ with $f(0) = 4$ and $f(h) = 200 \beta f(h-1)$ for $h \geq 1$, i.e., $f(h) = 4 \cdot (200 \beta)^h$.
\end{itemize}}

\begin{lemma}
\label{lem:RandTreeBal}
Consider an instance of \pbone that satisfies the assumptions in Theorem~\ref{thm:RandTreeBalMainLemma}. If $n^{10} \leq \Phi(t) \leq n^{20}$, then 
\[
\E[|L|] ~ \geq ~ \frac{1}{2 \cdot f(h)} \cdot \E[Q].
\]
Since $h \leq \log \log n/C$ for sufficiently large constant $C$, this implies that
\[
\E[\Delta \Phi(t)] \quad \leq \quad -\frac{\lambda}{2 \cdot f(h)} \E[Q] + \lambda^2 \E[Q] \quad \leq \quad 0.
\]
\end{lemma}


An important consequence of Lemma~\ref{lem:RandTreeBal} is that w.h.p. the potential never reaches a value $\geq n^{20}$ as there is negative drift towards zero  for $ \Phi(t) \in [n^{10}, n^{20}]$. This immediately implies Theorem~\ref{thm:RandTreeBalMainLemma}.

\begin{proofof}{Theorem~\ref{thm:RandTreeBalMainLemma}} 
Initially $\Phi(0) = n$. By Lemma~\ref{lem:RandTreeBal}, we have the following bounds on the change in $\Phi(t)$:  (1) When $\Phi(t) < n^{10}$, we have $\E[\Delta \Phi(t)] \leq \lambda^2 \Phi(t) < n^{10}$; (2) When $n^{10} \leq \Phi(t) \leq n^{20}$, we have $\E[\Delta  \Phi(t)] \leq 0$; (3) When $\Phi(t) > n^{20}$, change in potential $\Delta \Phi(t)$ can be arbitrary.

To handle Case (3), we define the following stochastic process $\widetilde{\Phi}(t)$: suppose $\widetilde{\Phi}(t)$ stays the same as $\Phi(t)$ before $\Phi(t)$ becomes larger than $n^{20}$. 
After the first $t$ where $\Phi(t) >n^{20}$, we set $\widetilde{\Phi}(t') = \Phi(t)$ for every $t' \geq t$, i.e., $\widetilde{\Phi}(t)$ stays fixed after time $t$. 
This means that whenever $\widetilde{\Phi}(t)$ exceeds $n^{20}$, we have $\Delta \widetilde{\Phi}(t') = 0$ for any $t' \geq t$.
Therefore, we always have $\E[\Delta \widetilde{\Phi}(t)]  \leq  n^{10}$,
which gives $\E[\widetilde{\Phi}(t)] \leq n^{11}$ for any $t \in [n]$.
Using Markov's inequality followed by a union bound, we have
\[
\p\Big[\exists t \in [n], \Phi(t) > n^{20} \Big] \quad = \quad \p\Big[\exists t \in [n], \widetilde{\Phi}(t) > n^{20}\Big] \quad \leq \quad \frac{1}{n^8}.
\]
Notice when $\Phi(t) \leq n^{20}$, the discrepancy of the tree $d(\T) = O(\log^2 (n)$,
which finishes the proof. 
\end{proofof}



\subsection{Proof of Lemma~\ref{lem:RandTreeBal} using a Separation Lemma}
\label{lem:LExceedsQ}

We define the following notion of a dangerous set for some value $x$. 
Essentially for any value $y$ in the dangerous set of $x$, hyperbolic sine $\sinh(\lambda y)$ will cancel out a significant fraction of $\sinh(\lambda x)$. 

\begin{defn}[Dangerous Set]
\label{defn:DangerousSet}
For any $x \in \mathbb{R}$ s.t. $|x| \geq \frac{\log n}{\lambda}$, we define the dangerous set of $x$ to be 
\[
\dang(x) ~ \overset{\Def}{=} ~ \left[-x - \frac{\log 10}{\lambda}, -x + \frac{\log 10}{\lambda} \right].
\]
\end{defn}

One immediate property is that $y \in \dang(x)$ is equivalent to $x \in \dang(y)$ if both $|x| \geq \frac{\log n}{\lambda}$ and $|y| \geq \frac{\log n}{\lambda}$. The following two facts follow immediately from Definition~\ref{defn:DangerousSet} and the properties of hyperbolic functions. We give their proofs in Appendix~\ref{sec:MissProofSec2}.
\begin{fact}
\label{claim:DangerousSetRatioBound}
For any $x \in \mathbb{R}$ s.t. $|x| \geq \frac{\log n}{\lambda}$ and $y \in \dang(x)$, we have
\[
\max \left\{ \frac{\cosh(\lambda x)}{\cosh(\lambda y)} ~,~ \frac{\cosh(\lambda y)}{\cosh(\lambda x)} \right\} ~ \leq ~ 11 \qquad \text{and} \qquad \max \left\{ \frac{|\sinh(\lambda x)|}{|\sinh(\lambda y)|} ~,~ \frac{|\sinh(\lambda y)|}{|\sinh(\lambda x)|} \right\} ~ \leq ~ 11.
\]
\end{fact}
On the other hand, if $y \notin \dang(x)$, then $\sinh(\lambda y)$ will cancel only a small portion of $\sinh(\lambda x)$.
\begin{fact}
\label{claim:DangerousSetNoCancel}
For any $x \in \mathbb{R}$ s.t. $|x| \geq \frac{\log n}{\lambda}$ and $y \notin \dang(x)$, we have
\IGNORE{
\[
\max \big\{|\sinh(\lambda x)|~,~ |\sinh(\lambda y)| \big\}  ~ \geq ~  9 \cdot \min\big\{|\sinh(\lambda x)|~,~ |\sinh(\lambda y)|\big\}.
\]
This implies that
}
\[
|\sinh(\lambda x)+\sinh(\lambda y)| ~ \geq ~ \frac{8}{9} \cdot \max \big\{|\sinh(\lambda x)|~,~ |\sinh(\lambda y)|\big\}.
\]
\end{fact}
To prove Lemma~\ref{lem:RandTreeBal}, we need the following lemma which forms the heart of our proof.

\IGNORE{\begin{lemma}[Separation Lemma]
\label{lem:ScalarSeparationLemma}
Consider any $m$-ary tree with height $h$ and fan-out $m \geq 100$. 
If $\Phi(t)\leq n^{20}$, then for any $x, \widetilde{d}_r \in \mathbb{R}$ s.t. $|x| \geq \frac{\log n}{\lambda} + \frac{h \log 10}{\lambda} $ and $\widetilde{d}_r \notin \dang(x)$, we have 
\begin{align}
\label{eqn:ScalarSeparationLemma}
\E \big[|\widetilde{L} + \sinh( \lambda x)| \big] ~ \geq ~ \frac{1}{f(h)} \cdot \E \big[\widetilde{Q} + \cosh(\lambda x)\big] - h n^2,
\end{align}
where $\widetilde{L}$ (resp. $\widetilde{Q}$) is obtained from $L$ (resp. $Q$) by replacing $d_r$ with $\widetilde{d_r}$, i.e.
\[
\widetilde{L} ~\overset{\Def}{=}~ \sinh(\lambda \widetilde{d_r}) + \sum_{v \in \calP_t \backslash \{r\}} \sinh(\lambda d_v) \qquad \text{and} \qquad  \widetilde{Q} ~\overset{\Def}{=} ~\cosh(\lambda \widetilde{d_r}) + \sum_{v \in \calP_t \backslash \{r\}} \cosh(\lambda d_v).
\]
\end{lemma}}

\begin{lemma}[Separation Lemma]
\label{lem:ScalarSeparationLemma}
Consider any $m$-ary tree with height $h$ and fan-out $m \geq 100$. For any $\widetilde{d}_r \in \mathbb{R}$, let  $\widetilde{L}$ (resp. $\widetilde{Q}$) be obtained from $L$ (resp. $Q$) by replacing $d_r$ with $\widetilde{d_r}$, i.e.
\[
\widetilde{L} ~\overset{\Def}{=}~ \sinh(\lambda \widetilde{d_r}) + \sum_{v \in \calP_t \backslash \{r\}} \sinh(\lambda d_v) \qquad \text{and} \qquad  \widetilde{Q} ~\overset{\Def}{=} ~\cosh(\lambda \widetilde{d_r}) + \sum_{v \in \calP_t \backslash \{r\}} \cosh(\lambda d_v).
\]
If  $\Phi(t)\leq n^{20}$, then for any $x$ such that $|x| \geq \frac{\log n}{\lambda} + \frac{h \log 10}{\lambda} $ and $\widetilde{d}_r \notin \dang(x)$,  we have 
\begin{align}
\label{eqn:ScalarSeparationLemma}
\E \big[|\widetilde{L} + \sinh( \lambda x)| \big] ~ \geq ~ \frac{1}{f(h)} \cdot \E \big[\widetilde{Q} + \cosh(\lambda x)\big] - h n^2,
\end{align}
\end{lemma}

\IGNORE{\begin{lemma}[Separation Lemma]
\label{lem:ScalarSeparationLemma}
Consider any $m$-ary tree with height $h$ and fan-out $m \geq 100$ where the imbalance $d_r$ of the root is replaced by some other value $\widetilde{d}_r$. If $\Phi(t)\leq n^{20}$ before the replacement, then for any $x \in \mathbb{R}$ s.t. $|x| \geq \frac{\log n}{\lambda} + \frac{h \log 10}{\lambda} $ and $\widetilde{d}_r \notin \dang(x)$, we have 
\begin{align}
\label{eqn:ScalarSeparationLemma}
\E \big[|\widetilde{L} + \sinh( \lambda x)| \big] ~ \geq ~ \frac{1}{f(h)} \cdot \E \big[\widetilde{Q} + \cosh(\lambda x)\big] - h n^2,
\end{align}
where $\widetilde{L}$ (resp. $\widetilde{Q}$) is obtained from $L$ (resp. $Q$) by replacing $d_r$ with $\widetilde{d_r}$, i.e.
\[
\widetilde{L} ~\overset{\Def}{=}~ \sinh(\lambda \widetilde{d_r}) + \sum_{v \in \calP_t \backslash \{r\}} \sinh(\lambda d_v) \qquad \text{and} \qquad  \widetilde{Q} ~\overset{\Def}{=} ~\cosh(\lambda \widetilde{d_r}) + \sum_{v \in \calP_t \backslash \{r\}} \cosh(\lambda d_v).
\]
\end{lemma}}

To understand the lemma statement, think of $|\sinh(\cdot)|$ as  the same as $\cosh(\cdot)$ (this only introduces  small errors) and $\widetilde{d}_r = d_r$ (i.e.,  the imbalance of the root is not replaced). Now the quantity inside $\E[\cdot]$ on the LHS of~(\ref{eqn:ScalarSeparationLemma}) can be written as the magnitude of the summation of several $\sinh(\cdot)$ terms. In general, these $\sinh(\cdot)$ terms might cancel each other and result in the summation being small. However, Lemma~\ref{lem:ScalarSeparationLemma} states that for a randomly sampled root-leaf path $\calP_t$, the magnitude of the summation roughly {\em separates} to become the summation of the magnitudes, i.e., $|\sum_{v \in \calP_t} \sinh(\lambda d_v)| \approx \sum_{v \in \calP_t} |\sinh(\lambda d_v)| \approx Q$. The terms $\sinh(\lambda x)$ and  $hn^2$ are due to technical reasons and are not important for understanding the lemma statement.

\begin{remark}
The exponential factor $f(h) = 2^{O(h)}$ in Lemma~\ref{lem:ScalarSeparationLemma} is  tight. In Section~\ref{sec:SepLemmaTightness} we show an example where $\E[|L|] < \exp(-\Omega(h)) \cdot \E[Q] - hn^2$.
\end{remark}

Before proving the Separation Lemma, we first finish the proof of Lemma~\ref{lem:RandTreeBal}.

\begin{proofof}{Lemma~\ref{lem:RandTreeBal} from Separation Lemma}
Since $\Phi(t) \geq n^{10}$ by assumption, we have $\E[Q] \geq n^9$ because the total number of nodes in the tree is at most $n$.
Now we pick any $x \in \left[\frac{\log n}{\lambda} + \frac{h \log 10}{\lambda}, \frac{2 \log n}{\lambda} \right]$ such that the root imbalance $d_r \notin \dang(x)$.  
This can always be done because $h = O(\log \log n)$ implies that $\frac{h \log 10}{\lambda} \ll \frac{\log n}{\lambda}$ and that the width of the dangerous set is $\frac{2 \cdot \log 10}{\lambda} \ll \frac{\log n}{\lambda}$.
Now applying Lemma~\ref{lem:ScalarSeparationLemma} with $\widetilde{d}_r = d_r$ (i.e., keeping the imbalance of the root), we have 
\begin{align}
\label{eqn:InductionHypothesisTreeBal}
\E\left[ |L + \sinh(\lambda x)| \right] ~ \geq ~ \frac{1}{f(h)} \cdot \E[Q + \cosh(\lambda x)] - h n^2.
\end{align}
Since $|x| < \frac{2 \log n}{\lambda}$, we have $|\sinh(\lambda x)| \leq n^2  \leq \E[Q]/n^7$. It follows that
\[
\E[|L|] ~\geq~ \E[|L + \sinh(\lambda x)|] - |\sinh(\lambda x)| ~\overset{(\ref{eqn:InductionHypothesisTreeBal})}{\geq}~ \frac{1}{f(h)} \cdot \E[Q + \cosh(\lambda x)] - hn^2 - |\sinh(\lambda x)| ~\geq~ \frac{1}{2  f(h)} \cdot \E[Q],
\]
which finishes the proof of Lemma~\ref{lem:RandTreeBal}.
\end{proofof}



\subsection{Proof of the Separation Lemma}

Before we formally prove the Separation Lemma, we give an overview of our proof strategy. Our plan is to induct on the height  $h$ of  tree $\T$. We call $x$ the {\em entering value} for convenience. For the induction basis $h=0$ where $\T$ contains a single node $r$, the assumption that $\widetilde{d}_r \notin \dang(x)$ allows us to {\em separate} $|\sinh(\lambda \widetilde{d}_r)+\sinh(\lambda x)|$ to $|\sinh(\lambda \widetilde{d}_r)|+|\sinh(\lambda x)|$. In the induction step, we would like to argue that $\widetilde{d_r} \notin \dang(x)$ allows us to view $\widetilde{d}_r$ and $x$ together as an entering value $x'$ for the subtrees, where $\sinh(\lambda x') = \sinh(\lambda \widetilde{d}_r) + \sinh(\lambda x)$, and apply induction hypothesis. Unfortunately, this only works for the set of  ``safe subtrees'' rooted at $s \in \child(r)$ that satisfy $d_s \notin \dang(x')$ as  required by the induction hypothesis. In general, this condition may not hold for every subtree.

To deal with this problem, we crucially use the fact that the imbalance of all children $s \in \child(r)$ sum up to $d_r$, which we know is small because $\Phi(t)$ is small by assumption. This implies that at least a constant fraction of the children of $r$ are safe. If either the entering value $x'$ or the set of safe subtrees are {``heavy''} (i.e., constitute a large fraction of  $\widetilde{Q}$), then we are done by directly applying induction hypothesis on the set of safe subtrees. 
If on the other hand both the entering value $x'$ and the set of safe subtrees are {``light''},  we have to separate the $\sinh(\cdot)$ terms even when we enter a ``dangerous subtree''  with $d_s \in \dang(x')$. Our approach is to combine the entering value $x'$ and $d_s$ to obtain $\widetilde{d}_s$, and then pick some small $x_s'$ such that $\widetilde{d}_s \notin \dang(x_s')$ as the new entering value. Since $x'$ is light and $x_s'$ is small, such modifications incur only small losses on both the $\widetilde{L}$ and $\widetilde{Q}$ terms. We can now apply the induction hypothesis using $\widetilde{d}_s$ with entering value $x_s'$ to separate the $\sinh(\cdot)$ terms even when we enter a dangerous subtree. This is the reason why the Separation Lemma allows the root imbalance $d_r$ to be replaced be any arbitrary value $\widetilde{d}_r$.


\begin{proofof}{Lemma~\ref{lem:ScalarSeparationLemma}} Since $\Phi(t) \leq n^{20}$, we know that $|d_v|  \leq  \frac{21 \log n}{\lambda}$ for all $v \in \T$ before the replacement of the root imbalance. Note that after replacing the root imbalance $d_r$ by $\widetilde{d}_r$, we might not have $|\widetilde{d}_r| \leq \frac{21 \log n}{\lambda}$. However, we still have
\begin{align}
\label{eqn:SumOfChildrenDiscSmall}
\Big| \sum_{v \in \child(r)} d_v \Big| ~ \leq ~ \frac{21 \log n}{\lambda}.
\end{align}

\medskip
\noindent \textbf{Induction basis for $\mathbf{h =0 }$.} 
In this case the tree $\T$ contains only one node which is the root $r$. Therefore we have $|\widetilde{L} + \sinh(\lambda x)| ~ = ~  |\sinh(\lambda \widetilde{d}_r) + \sinh(\lambda x)|$.
Since $\widetilde{d}_r \notin \dang(x)$, it follows from Fact~\ref{claim:DangerousSetNoCancel} that
\[
|\widetilde{L} + \sinh(\lambda x)| ~ \geq ~ \frac{8}{9} \cdot \frac{1}{2} \cdot \Big(|\sinh(\lambda \widetilde{d}_r)| + |\sinh(\lambda x)| \Big) ~\geq~  \frac{1}{4} \cdot \Big(\cosh(\lambda \widetilde{d}_r) + \cosh(\lambda x)\Big) ~= ~ \frac{1}{4} \cdot \E \big[\widetilde{Q} + \cosh(\lambda x)\big],
\]
where the second inequality uses $|\sinh(\lambda y)| \geq \cosh(\lambda y) - 1$ for any $y \in \mathbb{R}$ and the assumption in Lemma~\ref{lem:ScalarSeparationLemma} that $|x| \geq \frac{\log n}{\lambda} + \frac{h \log 10}{\lambda}$.
This finishes the proof of the induction basis.

\medskip
\noindent \textbf{Induction step from height $\mathbf{(h-1)}$ to $\mathbf{h}$.}
Let $x' \in \mathbb{R}$ be the value such that $\sinh(\lambda x') = \sinh(\lambda x) + \sinh(\lambda \widetilde{d}_r)$. Since $\widetilde{d}_r \notin \dang(x)$ and $|x| \geq \frac{\log n}{\lambda} + \frac{h \log 10}{\lambda}$ by assumption in Lemma~\ref{lem:ScalarSeparationLemma}, it follows from Fact~\ref{claim:DangerousSetNoCancel} that 
\begin{align}
|x'| ~ \geq ~ \frac{\log n}{\lambda} + \frac{(h-1) \cdot \log 10}{\lambda}  \qquad \text{and} \qquad \cosh(\lambda x') ~\geq~ \frac{7}{10} \cdot \Big(\cosh(\lambda x) + \cosh(\lambda \widetilde{d}_r)\Big). \label{eqn:RootReplacement}
\end{align}
Therefore, one can imagine replacing $\widetilde{d}_r$ of the root by $x'$ and it suffices to show that
\begin{align}
\label{eqn:IgnoreRoot}
\E[|L_s + \sinh(\lambda x')|] ~ \geq ~ \frac{2}{ f(h)} \cdot \E[Q_s + \cosh(\lambda x')] - h n^2,
\end{align}
where $L_s$ (resp. $Q_s$) denote the $L$ (resp. $Q$) term in the (random) subtree $\T_s$ rooted at $s \in \child(r) \cap \calP_t$, i.e. 
\[
L_s ~\overset{\Def}{=} \sum_{v \in \calP_t \cap \T_s} \sinh(\lambda d_v) \qquad \text{and} \qquad Q_s ~\overset{\Def}{=} \sum_{v \in \calP_t \cap \T_s} \cosh(\lambda d_v).
\]
Once~(\ref{eqn:IgnoreRoot}) is established, we immediately have
\[
\E[|\widetilde{L} + \sinh(\lambda x)|] ~=~ \E[|L_s + \sinh(\lambda x')|] ~\overset{(\ref{eqn:IgnoreRoot})}{\geq} \frac{2}{f(h)} \cdot \E[Q_s + \cosh(\lambda x')] - hn^2 \overset{(\ref{eqn:RootReplacement})}{\geq} \frac{1}{f(h)} \cdot \E[\widetilde{Q} + \cosh(\lambda x)] - hn^2,
\]
which finishes the proof of the induction step. The only thing left is to prove~(\ref{eqn:IgnoreRoot}). 

To prove~(\ref{eqn:IgnoreRoot}), we assume without loss of generality that $x' \geq 0$ (the case where $x' \leq 0$ is similar). Under this assumption, we have from~(\ref{eqn:RootReplacement}) that 
\begin{align}
\label{eqn:x'Large}
x' ~ \geq ~ \frac{\log n}{\lambda} + \frac{(h-1) \log 10}{\lambda}.
\end{align}
Recall that the subtree rooted at $s \in \child(r)$ is denoted by $\T_s$. Among all the children of the root $r$, we denote the set that are dangerous for $x'$ as $D_{h-1} \overset{\Def}{=} \big\{s \in \child(r): d_s \in \dang(x')\big\}$.
Denote the set of children of root $r$ that are not dangerous for $x'$ as $S_{h-1} = \child(r) \backslash D_{h-1}$.
We first argue that 
\begin{align}
\label{eqn:SafeSetIsLarge}
|D_{h-1}|  ~\leq~  0.99m.
\end{align}
Assume for the purpose of contradiction that~(\ref{eqn:SafeSetIsLarge}) doesn't hold. We notice that for each $s \in D_{h-1}$, $d_s \in \dang(x')$ implies that  
\begin{align}
\label{eqn:DangerousChildNegative}
d_s \quad \leq \quad - x' + \frac{\log 10}{\lambda} \quad \overset{(\ref{eqn:x'Large})}{\leq} \quad -\frac{\log n}{\lambda} - \frac{(h-2) \log 10}{\lambda}.
\end{align}
Observe that 
\begin{align}
\label{eqn:SumofChildrenLight}
\Big|\sum_{v \in D_{h-1}} d_v + \sum_{v \in S_{h-1}} d_v \Big| \quad = \quad \Big|\sum_{v \in \child(r)} d_v \Big| \quad \overset{(\ref{eqn:SumOfChildrenDiscSmall})}{\leq} \quad \frac{21 \log n}{\lambda} \quad \leq \quad \frac{0.21 m \log n}{\lambda},
\end{align}
where the last inequality follows from the assumption that $m \geq 100$.
This together with the assumption that~(\ref{eqn:SafeSetIsLarge}) doesn't hold imply that
\[
\Big|\sum_{v \in S_{h-1}} d_v \Big| ~\geq~ \Big|\sum_{v \in D_{h-1}} d_v \Big| - \Big|\sum_{v \in D_{h-1}} d_v + \sum_{v \in S_{h-1}} d_v \Big| ~\geq~ \frac{0.99 m \log n}{\lambda} - \frac{0.21 m \log n}{\lambda} ~=~ \frac{0.78 m \log n}{\lambda},
\]
where the second inequality follows from~(\ref{eqn:DangerousChildNegative}) and (\ref{eqn:SumofChildrenLight}).
Notice that the assumption that~(\ref{eqn:SafeSetIsLarge}) doesn't hold also implies that $|S_{h-1}| \leq 0.01m$, so it follows that there is a node $v \in S_{h-1}$ with 
$|d_v| \geq \frac{50 \log n}{\lambda}$.
But this implies that $\cosh(\lambda d_v)  \geq  n^{49}$,
which is a contradiction to the assumption in Lemma~\ref{lem:ScalarSeparationLemma} that $\Phi(t) \leq n^{20}$. This establishes~(\ref{eqn:SafeSetIsLarge}).

We use $\E_{s \sim D_{h-1}}[\cdot]$ to denote the expectation when $s \in \child(r)$ is sampled from $D_{h-1}$ uniformly at random. 
Similarly, $\E_{s \sim S_{h-1}}[\cdot]$ is used to denote the expectation when $s \in \child(r)$ is sampled from $S_{h-1}$ uniformly at random. 
We keep the notation $\E[\cdot]$ for the case where $s$ is sampled from $\child(r)$ uniformly at random.
Recall that $D_{h-1} \cup S_{h-1} = \child(r)$ and that $D_{h-1} \cap S_{h-1} = \emptyset$.
These imply that when $s \in \child(r)$ is chosen uniformly at random, we have
\begin{align}
\label{eqn:LsDecomposition}
\E[|L_s + \sinh(\lambda x')|] ~=~ \frac{|D_{h-1}|}{m} \cdot \E_{s \sim D_{h-1}}\Big[|L_s + \sinh(\lambda x')|\Big] + \frac{|S_{h-1}|}{m} \cdot \E_{s \sim S_{h-1}}\Big[|L_s + \sinh(\lambda x')|\Big].
\end{align}
\IGNORE{
and
\[
\E[|Q_s + \cosh(\lambda x')|] = \frac{|D_{h-1}|}{m} \cdot \E_{s \sim D_{h-1}}[|Q_s + \cosh(\lambda x')|] + \frac{|S_{h-1}|}{m} \cdot \E_{s \sim S_{h-1}}[|Q_s + \cosh(\lambda x')|].
\]
}
Now we consider three different cases.

\noindent \textbf{Case 1 (heavy root): $\cosh(\lambda x') \geq \frac{1}{\beta} \cdot \E[Q_s + \cosh(\lambda x')]$.} 
In this case we directly apply induction hypothesis with $\widetilde{d}_s = d_s$ (i.e. keeping the imbalance of $s$) on the set of safe children $s \in S_{h-1}$ and get
\begin{align}
\E[|L_s + \sinh(\lambda x')|] 
&\overset{(\ref{eqn:LsDecomposition})}{\geq} \frac{|S_{h-1}|}{m} \cdot \E_{s \sim S_{h-1}}[|L_s + \sinh(\lambda x')|] \nonumber \\
& \geq \frac{|S_{h-1}|}{m} \cdot \left(\frac{1}{f(h-1)} \cdot \E_{s \sim S_{h-1}}[Q_s + \cosh(\lambda x')] - (h-1) n^2  \right) \label{eqn:SafeSetInduct}
\end{align}
Since we have $\cosh(\lambda x') \geq \frac{1}{\beta} \cdot \E[Q_s + \cosh(\lambda x')]$ in this case, we can further simplify this as
\begin{align*}
\E[|L_s + \sinh(\lambda x')|]
&\overset{(\ref{eqn:SafeSetIsLarge})}{\geq} 0.01 \cdot \frac{1}{f(h-1)} \cdot \cosh(\lambda x') - (h-1) n^2  \\
&\geq \frac{1}{100 \cdot \beta f(h-1)} \cdot  \E[Q_s + \cosh(\lambda x')] - hn^2 ~=~ \frac{2}{f(h)} \cdot \E[Q_s + \cosh(\lambda x')] - hn^2.
\end{align*}

\noindent \textbf{Case 2 (heavy safe $Q_s$): $\frac{|S_{h-1}|}{m} \cdot \E_{s \sim S_{h-1}}[Q_s + \cosh(\lambda x')] \geq \frac{1}{100} \cdot \E[Q_s + \cosh(\lambda x')]$.}
In this case we again apply induction hypothesis with $\widetilde{d}_s = d_s$ on the set $S_{h-1}$ and get~(\ref{eqn:SafeSetInduct}). Now using the assumption that $\frac{|S_{h-1}|}{m} \cdot \E_{s \sim S_{h-1}}[Q_s + \cosh(\lambda x')] \geq \frac{1}{100} \cdot \E[Q_s + \cosh(\lambda x')]$, we have
\[
\E[|L_s + \sinh(\lambda x')|] ~\geq~ \frac{1}{100 \cdot f(h-1)} \cdot \E[Q_s + \cosh(\lambda x')] - hn^2 ~\geq~ \frac{2}{f(h)} \cdot \E[Q_s + \cosh(\lambda x')] - hn^2.
\]

\noindent \textbf{Case 3 (light root and light safe $Q_s$): $\frac{|S_{h-1}|}{m} \cdot \E_{s \sim S_{h-1}}[Q_s + \cosh(\lambda x')] < \frac{1}{100} \cdot \E[Q_s + \cosh(\lambda x')]$ and $\cosh(\lambda x') < \frac{1}{\beta} \cdot \E[Q_s + \cosh(\lambda x')]$.}
The first condition implies that 
\begin{align}
\label{eqn:HeavyDangerousSet}
\frac{|D_{h-1}|}{m} \cdot \E_{s \sim D_{h-1}}[Q_s + \cosh(\lambda x')] ~ \geq ~  \frac{99}{100} \cdot \E[Q_s + \cosh(\lambda x')].
\end{align}
In this case, we cannot directly apply our induction hypothesis on the set $s \in D_{h-1}$ without replacing $d_s$ because each $s \in D_{h-1}$ is dangerous for $x'$, i.e. $d_s \in \dang(x')$.
To circumvent this problem, we define $\widetilde{d}_s \in \mathbb{R}$ to be such that  $\sinh(\lambda \widetilde{d}_s) =  \sinh(\lambda d_s) + \sinh(\lambda x')$
for each $s \in D_{h-1}$ and we pick $x_s' \in \left[ \frac{\log n}{\lambda} + \frac{(h-1) \log 10}{\lambda}, \frac{2 \log n}{\lambda} \right]$ s.t. $\widetilde{d}_s \notin \dang(x_s')$.
We notce that 
\begin{align}
\label{eqn:Lightxs'}
|\sinh(\lambda x_s')|  ~\leq~  n^2.
\end{align}
Whenever $s \in D_{h-1}$ is selected by the random root-leaf path $\calP_t$, we replace $d_s$ by $\widetilde{d}_s$ and $x'$ by $x_s'$ and then apply induction hypothesis using $\widetilde{d}_s$ and $x_s'$.
We use $\widetilde{L}_s$ (resp. $\widetilde{Q}_s$) to denote the $L_s$ (resp. $Q_s$) term when $d_s$ is replaced by $\widetilde{d}_s$.
We have
\begin{align*}
\E[|L_s + \sinh(\lambda x')|] 
&\overset{(\ref{eqn:LsDecomposition})}{\geq} \frac{|D_{h-1}|}{m} \cdot \E_{s \sim D_{h-1}}[|L_s + \sinh(\lambda x')|] \\
&= \frac{|D_{h-1}|}{m} \cdot \E_{s \sim D_{h-1}}[|\widetilde{L}_s + \sinh(\lambda x_s') - \sinh(\lambda x_s')|] \overset{(\ref{eqn:Lightxs'})}{\geq} \frac{|D_{h-1}|}{m} \cdot \E_{s \sim D_{h-1}}[|\widetilde{L}_s + \sinh(\lambda x_s')|] - n^2.
\end{align*}
Now we apply induction hypothesis on each $s \in D_{h-1}$ after replacing $d_s$ by $\widetilde{d}_s$ and $x'$ by $x_s'$ to get
\begin{align*}
\E[|L_s + \sinh(\lambda x')|] 
&\geq\frac{|D_{h-1}|}{m} \cdot \frac{1}{f(h-1)} \cdot \E_{s \sim D_{h-1}}[\widetilde{Q}_s + \cosh(\lambda x_s')] - hn^2 \\
&\geq \frac{1}{f(h-1)} \cdot \frac{|D_{h-1}|}{m} \cdot \E_{s \sim D_{h-1}}[Q_s - \cosh(\lambda d_s)] - hn^2,
\end{align*}
where the second inequality follows from $\widetilde{Q}_s =  Q_s - \cosh(\lambda d_s)+ \cosh(\lambda \widetilde{d}_s)$.
Recall from Fact~\ref{claim:DangerousSetRatioBound} that $d_s \in \dang(x') $ implies that $\cosh(\lambda d_s) \leq 11 \cosh(\lambda x')$. This gives
\begin{align*}
\E[|L_s + \sinh(\lambda x')|]  
&\geq \frac{1}{f(h-1)} \cdot \frac{|D_{h-1}|}{m} \cdot \E_{s \sim D_{h-1}}[Q_s - 11 \cosh(\lambda x')] - hn^2 \\
&\geq \frac{1}{f(h-1)} \cdot \left( \frac{|D_{h-1}|}{m} \cdot \E_{s \sim D_{h-1}}[Q_s + \cosh(\lambda x')] - 12 \cosh(\lambda x') \right) - hn^2.
\end{align*}
From~(\ref{eqn:HeavyDangerousSet}) and the assumption that $\cosh(\lambda x') < \frac{1}{\beta} \cdot \E[Q_s + \cosh(\lambda x')]$, it follows that
\begin{align*}
\E[|L_s + \sinh(\lambda x')|]  
&\geq \frac{1}{f(h-1)}  \cdot \left(\frac{99}{100} \cdot \E[Q_s + \cosh(\lambda x')] - \frac{12}{\beta} \cdot \E[Q_s + \cosh(\lambda x')] \right) - hn^2\\
&\geq \frac{2}{f(h)} \cdot \E[Q_s + \cosh(\lambda x')] - hn^2,
\end{align*}
where the second inequality follows from our choice of $\beta = 100$. This finishes the proof of Lemma~\ref{lem:ScalarSeparationLemma}.
\end{proofof}

\IGNORE{
\begin{align*}
\E[|L_s + \sinh(\lambda x')|] 
&\overset{\cirt 2}{\geq} \frac{|D_{h-1}|}{m} \cdot \left(\frac{1}{f(h-1)} \cdot \E_{s \sim D_{h-1}}[Q_s' + \cosh(\lambda x_s')] - (h-1)n^2 \right) - n^2 \\
& \geq \frac{|D_{h-1}|}{m} \cdot \frac{1}{f(h-1)} \cdot \E_{s \sim D_{h-1}}[Q_s - \cosh(\lambda d_s)] - hn^2 \\
& \overset{\cirt 3}{\geq} \frac{1}{f(h-1)} \cdot \frac{|D_{h-1}|}{m} \cdot \E_{s \sim D_{h-1}}[Q_s - 10 \cosh(\lambda x')] - hn^2 \\
& \geq \frac{1}{f(h-1)}  \cdot \left(\frac{|D_{h-1}|}{m} \cdot \E_{s \sim D_{h-1}}[Q_s + \cosh(\lambda x')] - 11 \cosh(\lambda x') \right) - hn^2 \\
&\overset{\cirt 4}{\geq} \frac{1}{f(h-1)}  \cdot \left(\frac{99}{100} \cdot \E[Q_s + \cosh(\lambda x')] - \frac{11}{\beta} \cdot \E[Q_s + \cosh(\lambda x')] \right) - hn^2\\
&\geq \frac{2}{f(h)} \cdot \E[Q_s + \cosh(\lambda x')] - hn^2,
\end{align*}
where {\cirt 1} is because $|x_s'| \leq \frac{2 \log n}{\lambda}$, {\cirt 2} is by induction hypothesis, {\cirt 3} is because $d_s \in \dang(x')$, and {\cirt 4} is by the case assumptions that $\frac{|S_{h-1}|}{m} \cdot \E_{s \sim S_{h-1}}[Q_s + \cosh(\lambda x')] < \frac{1}{100} \cdot \E[Q_s + \cosh(\lambda x')]$ and $\cosh(\lambda x') < \frac{1}{\beta} \cdot \E[Q_s + \cosh(\lambda x')]$.
}


\section{\pbone} \label{sec:1DTusnady}
Recall that in the Online Interval Discrepancy problem, $n$ points arrive uniformly at random in the interval $[0,1]$, in an online manner. Upon $t$th arrival, we need to immediately  assign it a {\em color} $\col(t) \in \{-1,+1\}$.
Our goal is to minimize the discrepancy of the set system $\calS$ consisting of all sub-intervals of $[0,1]$, i.e.,
\[
\calS ~\overset{\Def}{=}~ \{[a,b]\}_{0\leq a < b \leq 1}.
\]
This section is devoted to proving Theorem~\ref{thm:pbone} which is restated as follows for reference. 

\thmpbone*


We prove Theorem~\ref{thm:pbone} via a reduction to  the \pbtreebalance problem from Section~\ref{sec:OnlineTreeBal}. This reduction appears in Section~\ref{subsec:TreeBaltoRand1DTus}, and in Section~\ref{sec:lowerbounds} we discuss lower bounds for this problem.


\subsection{From \pbtreebalance to \pbone}
\label{subsec:TreeBaltoRand1DTus}

\begin{proofof}{Theorem~\ref{thm:pbone}}
We consider the following embedding of the interval $[0,1]$ into a complete $m$-ary tree $\T$ with height $h = \log \log n/C$ for sufficiently large constant $C$ and fan-out $m = \lceil n^{\frac{1}{h+1}} \rceil \geq 100$. For simplicity, we assume $n^{\frac{1}{h+1}}$ is an integer and write $m = n^{\frac{1}{h+1}}$. The root $r$ corresponds to the interval $[0,1]$ and the $i$th node at depth $j$ corresponds to the interval $[ \frac{i-1}{m^j}, \frac{i}{m^j} ]$. 
Note that each depth-$j$ (assume $j < h$) interval is partitioned by $m$ depth-$(j+1)$ intervals which correspond to its $m$ children.
The set of leaves is formed by $m^h = n^{\frac{h}{h+1}}$ intervals each of length $\frac{1}{m^h}$.
For each node $v \in \T$, we also use $v$ to denote the corresponding interval.

Our algorithm for \pbone simply builds the above tree $\T$ and runs the algorithm for \pbtreebalance in Theorem~\ref{thm:RandTreeBalMainLemma}. We prove that w.h.p. this algorithm has discrepancy $O(n^{c/\log\log n})$ for the set system $\calS$, where $c$ is some universal constant.
Since each arrival lands uniformly at random in $[0,1]$, in order to bound the discrepancy of the set of all intervals, we only need to bound the discrepancy of the set of $O(n^2)$ intervals $[\frac{a}{n}, \frac{b}{n}]$ where $a, b \in \{0,1,\cdots,n\}$ and $a < b$. For each such interval $I = [\frac{a}{n}, \frac{b}{n} ]$, our tree embedding ensures that there's a collection $\B(I) = \{v_i\}_{i \in [s]}$ of at most $2mh$ nodes of $\T$ and two intervals $I_1 \subseteq l_1$ and $I_2 \subseteq l_2$, where $l_1$ and $l_2$ are two leaves of $\T$, such that (1) all these $2mh + 2$ intervals are disjoint, and (2) $I = \Big(\bigcup_{i \in [s]} v_i \Big) \cup I_1 \cup I_2$.
We can therefore bound the imbalance of $I$ as
\[
|\col(I)| ~\leq~ \sum_{i \in [s]} |\col(v_i)| + |\col(I_1)| + |\col(I_2)|.
\]
By Theorem~\ref{thm:RandTreeBalMainLemma}, w.h.p. each $v_i \in \T$ has 
\begin{align}
\label{eqn:TreeNodesBalance}
|\col(v_i)| ~= ~ O(\log^2 (n)).
\end{align}
Notice that each leaf $l$ of the tree has $n/m^h = m$ arrivals in expectation, so the two leaves $l_1$ and $l_2$ have at most $\widetilde{O}(m)$ arrivals w.h.p.. Therefore in this case we have 
\begin{align}
\label{eqn:LeavesFewArrivals}
|\col(I_1)| + |\col(I_2)| ~ = ~ \widetilde{O}(m).
\end{align} 
It follows from~(\ref{eqn:TreeNodesBalance}) and~(\ref{eqn:LeavesFewArrivals})  that w.h.p.,
\[
|\col(I)| \quad = \quad O(\log^2 (n)) + \widetilde{O}(m) \quad = \quad O(n^{c/\log \log n}),
\]
for some universal constant $c$. This finishes the proof of Theorem~\ref{thm:pbone}.
\end{proofof}

\subsection{Lower Bounds for \pbone}
\label{sec:lowerbounds}

\subsubsection{Adaptive Adversary}
When the arrival sequence is given by an \emph{adaptive adversary} (i.e., one who can decide the next arrival based on  the previous decisions of the algorithm) instead of being stochastic, the discrepancy can be $\Omega(n)$.

\begin{lemma}\label{lem:1DTusnadyAdapAdv}
For \pbone with an adaptive adversary, any online algorithm has discrepancy $\Omega(n)$.
\end{lemma}
\begin{proof}

Without loss of generality, assume the first two arrivals are at $0$ and $1$ and the algorithm colors them $+1$ and $-1$, respectively. 
Consider an adaptive adversary that always makes the next arrival land between the rightmost $+1$ element and the leftmost $-1$ element. This ensures that every element left of the last arrival has color $+1$ and every element right of the last arrival has color $-1$.
Therefore the discrepancy is at least $n/2$.
\end{proof}

\subsubsection{Oblivious Adversary}
The adaptive adversary in Lemma~\ref{lem:1DTusnadyAdapAdv} can be turned into an \emph{oblivious adversary} by making a random guess of the color of the last arrival.
\begin{lemma}\label{lem:1DTusnadyOblivAdv}
For \pbone with an oblivious adversary, any (randomized) online algorithm has (expected) discrepancy  $\Omega(\sqrt{n})$.
\end{lemma}
\begin{remark}
Recall that random coloring gives discrepancy $O(\sqrt{n \log n})$ against an oblivious adversary. Lemma~\ref{lem:1DTusnadyOblivAdv} implies that random coloring is tight up to a $O(\sqrt{\log n})$ factor. This shows that the assumption of stochastic arrivals is crucial for obtaining sub-polynomial discrepancy bounds.
\end{remark}

\begin{proofof}{Lemma~\ref{lem:1DTusnadyOblivAdv}}
Consider a randomized oblivious adversary that produces the following distribution of arrival sequences. The oblivious adversary first guesses a sequence of colors  $\widetilde{\col} = (\widetilde{\col}(1),\cdots,\widetilde{\col}(n))$ where each $\widetilde{\col}(i)\in \{-1,+1\}$ is chosen independently and uniformly at random. It then generates the arrival sequence using the adaptive adversary in Lemma~\ref{lem:1DTusnadyAdapAdv} for coloring $\widetilde{\col}$. 
More specifically, the oblivious adversary makes the $i$th arrival land between the rightmost $+1$ element and the leftmost $-1$ element among the first $(i-1)$ elements under the coloring $\widetilde{\col}(1),\cdots, \widetilde{\col}(i-1)$. The following observation is crucial to our analysis.
\begin{obs}
\label{obs:RandomGuessColor}
The position of the $i$th arrival doesn't depend on the random guess $\widetilde{\col}(i)$. Therefore, upon seeing only the first $i$ arrivals, an algorithm obtains no information about $\widetilde{\col}(i)$.
\end{obs}
Notice that for any randomly chosen $\widetilde{\col}$,
all arrivals $i$ with $\widetilde{\col}(i)= +1$ are on the left of all arrivals $j$ with $\widetilde{\col}(j) = -1$, so the discrepancy of the set $\calS$ of all intervals with respect to the coloring $\widetilde{\col}$ is at least $\Omega(n)$ according to Lemma~\ref{lem:1DTusnadyAdapAdv}.

However, the randomly guessed coloring $\widetilde{\col}$ might be different from the coloring $\col$ assigned by the online algorithm.
Nevertheless, Observation~\ref{obs:RandomGuessColor} implies that with probability at least $1/4$, the oblivious adversary makes at least $n/2 + \sqrt{n}/100$ correct guesses, i.e.
\[
\p\left[\big|\{i \in [n]: \widetilde{\col}(i) = \col(i)\} \big| \geq n/2 + \sqrt{n}/100 \right] ~\geq~ 1/4.
\]
We show in the following that whenever such an event happens, the discrepancy of the online algorithm's coloring is at least $\Omega(\sqrt{n})$. 

To prove this, we first define the following notations. We use $S_+$ (resp. $S_-$) to denote the set of arrivals with guessed color $+1$ (resp. $-1$) in $\widetilde{\col}$, i.e.
\[
S_+ ~\overset{\Def}{=}~ \{i \in [n]: \widetilde{\col}(i) = +1\} \qquad \text{and} \qquad S_- ~\overset{\Def}{=}~ \{i \in [n]: \widetilde{\col}(i) = -1\}.
\]
We use $n_+$ (resp. $n_-$) to denote the size of $S_+$ (resp. $S_-$), i.e.
\[
n_+ ~\overset{\Def}{=} ~ |S_+| \qquad \text{and} \qquad n_- ~\overset{\Def}{=}~ |S_-|.
\]
If the oblivious adversary makes at least $n/2 + \sqrt{n}/100$ correct guesses, then he makes either at least $(n_+/2 + \sqrt{n}/200)$ correct guesses among $S_+$ or at least $(n_-/2 + \sqrt{n}/200)$ correct guesses among $S_-$. Without loss of generality, we assume he makes at least $(n_+/2 + \sqrt{n}/200)$ correct guesses among arrivals in $S_+$ (the other case is similar). From this we have $\big|\{i \in S_+: \col(i) = +1\} \big| ~ \geq ~ n_+/2 + \sqrt{n}/200$.
It immediately follows that $\big|\{i \in S_+: \col(i) = -1\}\big| ~ \leq ~ n_+/2 - \sqrt{n}/200$.
Since $S_+$ is on the left of $S_-$, there's an interval $I_+$ that contains all arrivals in $S_+$. The imbalance of $I_+$ under the coloring $\col$ is at least
\[
|\col(I_+)| \quad = \quad \Big| |\{i \in S_+: \col(i) = +1\}| - |\{i \in S_+: \col(i) = -1\}|  \Big| \quad \geq \quad \sqrt{n}/100.
\]
We conclude that for any (randomized) online algorithm, when the input is drawn from such a distribution by the oblivious adversary, the discrepancy of the algorithm is at least $\Omega(\sqrt{n})$ with probability at least $1/4$. 
This also implies that for any (randomized) online algorithm, there's a {\em deterministic} oblivious adversary input such that the expected discrepancy of the algorithm is $\Omega(\sqrt{n})$.
Therefore no (randomized) online algorithm can achieve expected discrepancy better than $\Omega(\sqrt{n})$ against an oblivious adversary.
\end{proofof}

\subsubsection{Stochastic Arrivals}

Next we use the oblivious adversary lower bound in Lemma~\ref{lem:1DTusnadyOblivAdv} to prove a super-constant lower bound for \pbone with stochastic arrivals. 
\begin{lemma}
For \pbone with stochastic arrivals, any (randomized) online algorithm has expected discrepancy at least $\Omega((\log n)^{1/4})$.
\end{lemma}
\begin{proof}
Suppose $N = \frac{\sqrt{\log n}}{10}$. For any online algorithm, we know from Lemma~\ref{lem:1DTusnadyOblivAdv} that there's a deterministic oblivious adversary input of length exactly $N$ such that the online algorithm has discrepancy $\Omega(\sqrt{N})$. We show that with constant probability, there exists a sub-interval with such an oblivious adversary input of length exactly $N$.

We divide the unit interval into $n$ disjoint pieces, each of length $1/n$. 
Divide every piece further into $2^N$ sub-pieces, each of length $\epsilon$ where $\epsilon \cdot 2^N  = 1/n$, or equivalently $\epsilon n = 2^{-N}$. 
For each $i \in [n]$, we denote $X_i$ the indicator random variable that piece $i$ has exactly $N$ arrivals, i.e.
\[
X_i ~\overset{\Def}{=}~ \begin{cases}
1 \qquad &\text{if piece $i$ has exactly $N$ arrivals} \\
0 &\text{otherwise}
\end{cases}
\]
and denote $X \overset{\Def}{=} \sum_{i \in [n]} X_i$ the total number of pieces that has $N$ arrivals.
For any piece $i \in [n]$, the probability that piece $i$ has exactly $N$ arrivals is
\[
p ~\overset{\Def}{=} ~{n \choose N} \cdot \left( \frac{1}{n} \right)^N \cdot \left(1 - \frac{1}{n} \right)^{n-N} ~ = ~ \widetilde{\Omega}\left(\frac{1}{N^N} \right).
\]
Therefore the expected number of pieces with exactly $N$ arrivals is 
\[
\E[X] \quad=\quad np \quad=\quad \widetilde{\Omega}\left(\frac{n}{N^N} \right) \quad \geq \quad n^{0.9}. 
\]
We define event $\A = \{X \geq n^{0.8}\}$ to be the event that  the number of pieces with exactly $N$ arrivals is at least $n^{0.8}$.
We show in the following that
\begin{align}
\label{eqn:ManyPieceHasNArrivals}
\p[\A] ~\geq~ 0.1.
\end{align}

To prove~(\ref{eqn:ManyPieceHasNArrivals}), we need the Paley-Zygmund theorem which states that for any random variable $X \geq 0$ with finite variance and any $0 \leq s \leq 1$, 
\[
\p\left[X > s \E[X] \right] ~\geq~ (1-s^2) \E[X]^2/\E[X^2].
\]
Notice that all the $X_i$'s are identically distributed but not independent. We can calculate $\E[X]^2 / \E[X^2]$ as
\begin{align*}
\frac{\E[X]^2}{\E[X^2]} 
&= \frac{\sum_{i \in [n]} \E[X_i]^2 + \sum_{i \neq j} \E[X_i]\cdot \E[X_j]}{ \sum_{i \in [n]} \E[X_i^2] + \sum_{i \neq j} \E[X_i X_j]}\\
&= \frac{n p^2 + n(n-1) p^2}{np + n(n-1) \E[X_1X_2]} \quad = \quad \frac{n p^2 + n(n-1) p^2}{np + n(n-1) p \cdot \p\left[ X_2 = 1|X_1 = 1 \right]}.
\end{align*}
Notice that
\[
\p\left[ X_2 = 1|X_1 = 1 \right] \quad = \quad {n-N \choose N} \cdot \left( \frac{1}{n-1} \right)^N \cdot \left(1 - \frac{1}{n-1} \right)^{n-N} \quad \leq \quad 1.1 p.
\]
Therefore, we have
\[
\frac{\E[X]^2}{\E[X^2]} \quad \geq \quad \frac{n p^2 + n(n-1) p^2}{np + 1.1n(n-1) p^2} \quad \geq \quad 0.8.
\]
It follows from Paley-Zygmund theorem that 
\begin{align*}
\p \left[\A \right] \quad \geq \quad \p\left[ X \geq 0.1 \E[X]  \right] \quad \geq \quad (1-0.1)^2 \cdot \frac{\E[X]^2}{\E[X^2]} \quad \geq \quad 0.1,
\end{align*}
which finishes the proof of~(\ref{eqn:ManyPieceHasNArrivals}).

Now we condition on  event $\A$. For each piece that has exactly $N$ arrivals, the probability that it has the oblivious input from Lemma~\ref{lem:1DTusnadyOblivAdv} with length exactly $N$ is
\[		 
\left(\frac{\epsilon}{1/n} \right)^N \quad = \quad 2^{-N^2} \quad = \quad n^{-0.01},
\]
where we view each sub-piece of length $\epsilon$ as a discrete point.
Notice that condition on the number of arrivals in each piece, the arrival positions in each piece are independent of the arrival positions in other pieces. So condition on the number of arrivals in each piece such that event $\A$ holds, at least one piece has the random oblivious input of length $N$ w.h.p.. 

This combined with~(\ref{eqn:ManyPieceHasNArrivals}) imply that with constant probability, at least one piece $i$ has the oblivious adversary input of length exactly $N$. From Lemma~\ref{lem:1DTusnadyOblivAdv}, the algorithm has expected imbalance at least $\Omega(\sqrt{N})$ on the piece $i$. Therefore, any (randomized) online algorithm has expected discrepancy $\Omega((\log n)^{1/4})$ for \pbone with stochastic arrivals.
\end{proof}

\IGNORE{
Now the probability that a given piece has the oblivious adversary input from Lemma~\ref{lem:1DTusnadyOblivAdv}  of length exactly $N$ is 
\[		 
{n \choose N} \cdot \epsilon^N \cdot \Big(1 - \frac{1}{n} \Big)^{n-N} \quad = \quad \tilde{\Omega}(n \epsilon)^N \quad = \quad \tilde{\Omega}(2^{-N^2} ).
\]
Moreover, the expected number of  pieces with exactly $N$ arrivals is $\tilde{\Omega}(\frac{n}{N!})$.  Hence there is a constant probability of having an interval with oblivious adversary input of length $N$.
}


\section{\pbtwo}
\label{sec:OnlineRand2DStripe}

Recall that in the \pbtwo problem there are $n$ online arrivals where each one lands uniformly at random in the unit square $[0,1]\times [0,1]$.
Upon seeing arrival $t$, we need to immediately assign it a color  $\col(t) \in \{-1,+1\}$.
Our goal is to minimize the discrepancy of the following set  of \emph{stripes}:
\[
\calS ~=~ \Big\{[0,1] \times [a,b] \Big\}_{0 \leq a < b \leq 1} \cup \Big\{[a,b] \times [0,1] \Big\}_{0 \leq a < b \leq 1}.
\]

Our main result of this section is the following theorem.

\thmpbtwo*
\subsection{Our Algorithm}
We build two $m$-ary trees $\T_x$ and $\T_y$ of height $h = \log \log n/C$ for sufficiently large constant $C$, one for each axis, after projecting the square $[0,1] \times [0,1]$ to its corresponding axis. For each axis, the construction of the tree is the same as  in Section~\ref{sec:OnlineTreeBal}. Note that the roots of both the trees  correspond to the entire $[0,1]\times [0,1]$ square. Therefore, we use $r$ to denote roots of both the trees. For $\lambda \overset{\Def}{=} \frac{1}{\log n}$, define the  corresponding potential functions  as
\[
\Phi_x(t) ~ \overset{\Def}{=} ~ \sum_{v \in \T_x} \cosh(\lambda d_v)
\qquad \text{and} \qquad
\Phi_y(t) ~ \overset{\Def}{=} ~  \sum_{v \in \T_y} \cosh(\lambda d_v).
\]
The algorithm considers the potential function 
\[
\Phi(t) ~ \overset{\Def}{=} ~ \Phi_x(t) + \Phi_y(t)
\]
and simply assigns $\col(t) \in \{-1,+1\}$ to the $t$th arrival to minimize the  increase in the potential function.

By using the reduction from Section~\ref{subsec:TreeBaltoRand1DTus}, we show that Theorem~\ref{thm:StripeProblem} is implied by the following lemma. 

\IGNORE
{\begin{proofof}{Theorem~\ref{thm:2DStripeMain}}
Due to symmetry of the two axis, we only need to argue about stripes of the form $[a,b] \times [0,1]$ for $0 \leq a < b \leq 1$.
We are essentially projecting the square $[0,1] \times [0,1]$ to the $x$-axis and we can use the same argument as in the proof of Theorem~\ref{thm:OnlineTreeBalto1DTusnady} to finish the proof.
\end{proofof}
}

\begin{lemma}
\label{lem:2DStripePotentialBound}
The above potential-minimization algorithm satisfies that w.h.p.,
\[
d(\T_x) + d(\T_y) ~=~ O( \log^2 (n)).
\]
\end{lemma}

\begin{proofof}{Theorem~\ref{thm:StripeProblem}}
Define the set of $x$-stripes $\calS_x$ and the set of $y$-stripes $\calS_y$ as
\[
\calS_x ~\overset{\Def}{=}~ \Big\{[a,b] \times [0,1] \Big\}_{0 \leq a < b \leq 1} \qquad \text{and} \qquad \calS_y ~\overset{\Def}{=}~ \Big\{[0,1] \times [a,b] \Big\}_{0 \leq a < b \leq 1}.
\] 
By symmetry, we only need to argue that w.h.p., the set of $x$-stripes $\calS_x$ has discrepancy $O(n^{c/\log \log n})$ for some universal constant $c$. 
We project the unit square $[0,1] \times [0,1]$ onto the $x$-axis (equivalently, we only look at the $x$-coordinate of each arrival). Under such a projection, each stripe in $\calS_x$ becomes an interval $[a,b] \subseteq [0,1]$, and each arrival lands uniformly at random in the interval $[0,1]$. This reduces to the \pbone problem and the same argument in Section~\ref{subsec:TreeBaltoRand1DTus} proves that the discrepancy of $\calS_x$ is $O(n^{c/\log \log n})$ w.h.p. for some universal constant $c$.
\end{proofof}

The remaining section proves the missing Lemma~\ref{lem:2DStripePotentialBound}.

\subsection{Proof of Lemma~\ref{lem:2DStripePotentialBound} }
Denote $\calP_{x}(t)$ and $\calP_{y}(t)$ the randomly sampled root-leaf paths in $\T_x$ and $\T_y$ corresponding to the $t$th arrival. 
Notice that $\calP_x(t)$ and $\calP_y(t)$ are sampled independently and uniformly at random.
Recall from Section~\ref{sec:OnlineTreeBal},
\begin{align*}
\Delta \Phi(t) &\leq \lambda \Big(\sum_{v \in \calP_{x}(t)} \sinh(\lambda d_v) + \sum_{v \in \calP_{y}(t)} \sinh(\lambda d_v) \Big) \cdot \col(t) + \lambda^2 \Big(\sum_{v \in \calP_{x}(t)} \cosh(\lambda d_v) + \sum_{v \in \calP_{y}(t)} \cosh(\lambda d_v)  \Big)\\
&\overset{\Def}{=} \lambda (L_x + L_y) \cdot \col(t) + \lambda^2 (Q_x + Q_y)\\
&\leq -\lambda |L_x + L_y| + \lambda^2 (Q_x + Q_y).
\end{align*}

Now similar to the proof of Theorem~\ref{thm:RandTreeBalMainLemma} from Lemma~\ref{lem:RandTreeBal} for \pbtreebalance, to prove Lemma~\ref{lem:2DStripePotentialBound}  it is sufficient to prove the following Lemma~\ref{lem:2DStripePotentialIncrease}. 

\begin{lemma}
\label{lem:2DStripePotentialIncrease}
Consider the potential-minimization algorithm for \pbtwo. If $n^{10} \leq \Phi(t) \leq n^{20}$, then we have 
\[
\E[|L_x + L_y|] ~ \geq ~ \frac{1}{\beta^2 f(h)} \cdot \E[Q_x + Q_y].
\]
Since $h = \log \log n/C$ for sufficiently large constant $C$, this implies that
\[
\E[\Delta \Phi(t)] \quad \leq \quad -\lambda \cdot \E[|L_x + L_y|] + \lambda^2 \cdot \E[Q_x+Q_y] \quad \leq \quad 0.
\]
\end{lemma}

\begin{proofof}{Lemma~\ref{lem:2DStripePotentialBound} from Lemma~\ref{lem:2DStripePotentialIncrease}}
Initially $\Phi(0) = n$. By Lemma~\ref{lem:2DStripePotentialIncrease}, we have the following bounds on the change in $\Phi(t)$:  (1) When $\Phi(t) < n^{10}$, we have $\E[\Delta \Phi(t)] \leq \lambda^2 \Phi(t) < n^{10}$; (2) When $n^{10} \leq \Phi(t) \leq n^{20}$, we have $\E[\Delta  \Phi(t)] \leq 0$; (3) When $\Phi(t) > n^{20}$, change in potential $\Delta \Phi(t)$ can be arbitrary.

To handle Case (3), we define the following stochastic process $\widetilde{\Phi}(t)$: suppose $\widetilde{\Phi}(t)$ stays the same as $\Phi(t)$ before $\Phi(t)$ becomes larger than $n^{20}$. 
After the first $t$ where $\Phi(t) >n^{20}$, we set $\widetilde{\Phi}(t') = \Phi(t)$ for every $t' \geq t$, i.e., $\widetilde{\Phi}(t)$ stays fixed after time $t$. 
This means that whenever $\widetilde{\Phi}(t)$ exceeds $n^{20}$, we have $\Delta \widetilde{\Phi}(t') = 0$ for any $t' \geq t$.
Therefore, we always have $\E[\Delta \widetilde{\Phi}(t)]  \leq  n^{10}$,
which gives $\E[\widetilde{\Phi}(t)] \leq n^{11}$ for any $t \in [n]$.
Using Markov's inequality followed by a union bound, we have
\[
\p\Big[\exists t \in [n], \Phi(t) > n^{20} \Big] \quad = \quad \p\Big[\exists t \in [n], \widetilde{\Phi}(t) > n^{20}\Big] \quad \leq \quad \frac{1}{n^8}.
\]
Notice when $\Phi(t) \leq n^{20}$, the discrepancy of the trees $\T_x$ and $\T_y$ satisfy $d(\T_x) = O(\log^2 (n))$ and $d(\T_y) = O(\log^2 (n))$.
This finishes the proof of Lemma~\ref{lem:2DStripePotentialBound}.
\end{proofof}

\begin{proofof}{Lemma~\ref{lem:2DStripePotentialIncrease}}
Each arrival $t$ corresponds to picking independently a random root leaf path $\calP_{x}(t)$ from $\T_x$ and $\calP_{y}(t)$ from $\T_y$. For simplicity, we will denote them as  $\calP_x$ and $\calP_y$ when $t$ is clear from the context.
We assume without loss of generality that $\E[Q_x] \leq \E[Q_y]$ for arrival $t$ and we call $\T_x$ the \emph{lighter} tree.
Recall that the roots in $\T_x$ and $\T_y$  correspond to entire $[0,1]\times [0,1]$ square. Therefore, we use $r$ to denote both roots.
We consider two different cases.

\medskip
\textbf{Case 1 (light root): $\cosh(\lambda d_r) < \frac{1}{\beta \cdot f(h)} \cdot \E_{x,y}[Q_x + Q_y]$.} 
We consider $L_x$ for the randomly sampled root-leaf path $\calP_x$ of the lighter tree $\T_x$. 
Define $q_x$ to be such that $L_x = \sinh(\lambda q_x)$.
The idea is to bound $L_x + L_y = \sinh(\lambda q_x) + L_y$  using our Separation lemma in Lemma~\ref{lem:ScalarSeparationLemma}.
However, we cannot directly apply Lemma~\ref{lem:ScalarSeparationLemma} for $q_x$ and the tree $\T_y$ because of two reasons: (1) we might have $|q_x|<\frac{\log n}{\lambda} + \frac{h \log 10}{\lambda}$, or (2) $d_r$ might be dangerous for $q_x$.
The following modifications allow us to overcome these problems.

We first check condition (1). If true, then notice that $\E[Q_x+Q_y]\geq \Phi(t)/n \geq n^{9}$ so $\cosh(\lambda q_x) < n \ll \E[Q_x + Q_y]$.
So we can simply replace $q_x$ by some other discrepancy value $q_x' \in \left[\frac{\log n}{\lambda} + \frac{h \log 10}{\lambda}, \frac{2\log n}{\lambda} \right]$ such that $d_r \notin \dang(q_x')$.
Intuitively, this is fine because changing  $q_x$ to $q_x'$ only affects $|L|$ and $Q$ by a negligible amount. Moreover, if condition (1) is not true then set $q_x' = q_x$.

We then check condition~(2). If true, we simply ignore the root, i.e., replace $d_r$ by 0. Notice that after changing $q_x$ to $q_x'$ in the previous step so that condition~(1) is satisfied, we have $0 \notin \dang(q_x')$. Therefore, replacing $d_r$ by 0 will satisfy condition~(2). Intuitively, this change incurs only a tiny loss due to the case assumption that the root is ``light''.

After the above modifications, we can apply Lemma~\ref{lem:ScalarSeparationLemma} to get 
\begin{align}\label{eq:lightRootStrip}
\E_y [|(L_y - \sinh(\lambda d_r)) + \sinh(\lambda q_x')|] \geq \frac{1}{f(h)} \cdot \E_y[(Q_y - \cosh(\lambda d_r))+ \cosh(\lambda q_x')] - hn^2.
\end{align}
So we get $\E_y[|L_x + L_y|] = \E_y [|\sinh(\lambda q_x) + L_y|]$ is at least
\begin{align*}
& \qquad \E_y\Big[|L_y - \sinh(\lambda d_r) + \sinh(\lambda q_x')|\Big] - |\sinh(\lambda q_x') - \sinh(\lambda q_x)| - |\sinh(\lambda d_r)|\\
& \overset{\eqref{eq:lightRootStrip}}{\geq} \frac{1}{f(h)} \cdot \E_y \Big[Q_y - \cosh(\lambda d_r)+ \cosh(\lambda q_x')\Big] - |\sinh(\lambda q_x') - \sinh(\lambda q_x)| - |\sinh(\lambda d_r)| - hn^2.
\end{align*}
Since $\cosh(\lambda d_r) \geq |\sinh(\lambda d_r)|$ and $f(h) \geq 1$, we get 
\[ \E[|L_x + L_y|] \geq   \frac{1}{f(h)} \cdot \E\Big[Q_y + \cosh(\lambda q_x')\Big] - |\sinh(\lambda q_x') - \sinh(\lambda q_x)| - 2|\cosh(\lambda d_r)| - hn^2.
\]
After rearranging, this implies $\E_y[|L_x + L_y|]$ is at least
\[  \frac{1}{f(h)} \cdot \E_y\Big[Q_y + \cosh(\lambda q_x)\Big]  - | \cosh(\lambda q_x') -  \cosh(\lambda q_x)| - |\sinh(\lambda q_x') - \sinh(\lambda q_x)| - 2 |\cosh(\lambda d_r)| - hn^2.
\]
Now using that the root is light, i.e. $\cosh(\lambda d_r) < \frac{1}{\beta \cdot f(h)} \cdot \E_{x,y}[Q_x + Q_y]$, and that either $\cosh(\lambda q_x)$ equals $\cosh(\lambda q_x')$ or both are less than $n^2$ (similarly,  $\sinh(\lambda q_x)$ equals $\sinh(\lambda q_x')$ or both are less than $n^2$), we get 
\begin{align*}
\E_{x,y}[|L_x + L_y|] &\geq \frac{1}{f(h)} \cdot \E_{x,y}\Big[Q_y + \cosh(\lambda q_x)\Big] - (h+4)n^2 - \frac{2}{\beta f(h)} \cdot \E[Q_x + Q_y] \\
&\geq \frac{1}{\beta^2 f(h)} \cdot \E[Q_x + Q_y],
\end{align*}
where in the last inequality, we use the fact that $\T_x$ is lighter than $\T_y$, i.e. $\E[Q_x] \leq \E[Q_y]$.
This finishes the proof in this case.

\IGNORE{
\begin{align*}
&\geq \frac{1}{f(h)} \cdot \E[Q_y + \cosh(\lambda q_x)] - 2(|\sinh(\lambda q_x') - \sinh(\lambda q_x)| + |\sinh(\lambda d_r)|) - hn^2\\
&\geq \frac{1}{f(h)} \cdot \E[Q_y + \cosh(\lambda q_x)] - (h+4)n^2 - \frac{2}{\beta f(h)} \cdot \E[Q_x + Q_y] \\
&\geq \frac{1}{\beta^2 f(h)} \cdot \E[Q_y + \cosh(\lambda q_x)],
\end{align*}
so we are done in this case.
}

\medskip
\textbf{Case 2 (heavy root): $\cosh(\lambda d_r) \geq \frac{1}{\beta f(h)} \cdot \E[Q_x + Q_y]$.}
We define $q_x$ to be such that $\sinh(\lambda q_x) = L_x$ and define $p_x$ to be the probability that $q_x \in \dang(d_r)$.
Notice that 
\[ \cosh(\lambda d_r) \quad \geq \quad \frac{1}{\beta f(h)} \cdot \E[Q_x + Q_y] \quad \geq \quad \frac{1}{\beta f(h)} \cdot n^9, \] 
which implies that $d_r > \frac{2\log n}{\lambda}$. Thus the dangerous set for $d_r$ is well-defined.
We similarly define $p_y$ as the probability that $q_y \in \dang(d_r)$ where $q_y$ is such that $\sinh(\lambda q_y) = L_y$.
Now we consider two different sub-cases.

\medskip
\qquad \textbf{(i) $p_x \leq 0.9$ or $p_y \leq 0.9$.} Suppose $p_y \leq 0.9$, then we know that with at least  $0.1$ probability $q_y \notin \dang(d_r)$. 
We cannot yet apply our separation lemma in Lemma~\ref{lem:ScalarSeparationLemma} for such $q_y$ because it might be possible that $|q_y| < \frac{\log n}{\lambda} + \frac{h \log 10}{\lambda}$.
If that is the case, then we can simply replace $q_y$ by some $q_y' \in \left[\frac{\log n}{\lambda} + \frac{h \log 10}{\lambda}, \frac{2\log n}{\lambda} \right]$ such that $q_y' \notin \dang(d_r)$.
Notice that such a replacement affects both the $L$ term and $Q$ term by an additive change of at most $n^2$ which is much smaller than $\frac{1}{\beta^2 f(h)} \cdot \E[Q_x + Q_y]$.
We can therefore assume (up to additive $O(n^2)$ error) that whenever $q_y \notin \dang(d_r)$, we have $|q_y| \geq \frac{\log n}{\lambda} + \frac{h \log 10}{\lambda}$.
Now we can apply our separation lemma in Lemma~\ref{lem:ScalarSeparationLemma} for the path $\calP_y$ s.t. $d_r \notin \dang(q_y)$, we have
\[
\E[|L_x+L_y|] \quad \geq \quad 0.1 \cdot \Big( \frac{1}{f(h)} \cdot \cosh(\lambda d_r) - hn^2 \Big)\quad \geq \quad \frac{1}{\beta^2 f(h)} \cdot \E[Q_x + Q_y],
\]
so we are done in this case. The other situation where $p_x \leq 0.9$ is similar.

\medskip
\qquad \textbf{(ii) Both $p_x > 0.9$ and $p_y > 0.9$.} In this case, since $\calP_x$ and $\calP_y$ are independent, with probability at least $0.81$ we have $d_r \in \dang(q_x) \cap \dang(q_y)$. 
Since $\cosh(\lambda d_r) \geq \frac{1}{\beta f(h)} \cdot \E[Q_x + Q_y]$, we have $|d_r| \geq \frac{5 \log n}{\lambda}$. From Fact~\ref{claim:DangerousSetRatioBound}, this implies that $\cosh(\lambda q_x) \geq \frac{1}{11}  \cosh(\lambda d_r)$ and $\cosh(\lambda q_y) \geq \frac{1}{11}  \cosh(\lambda d_r)$, and that $q_x$ and $q_y$ have the same sign.
So we have with probability at least $0.81$,
\[
|L_x + L_y| ~ = ~ |\sinh(\lambda q_x)| + |\sinh(\lambda q_y)| ~> ~\frac{1}{11} \cdot \cosh(\lambda d_r) ~ \geq ~\frac{1}{11 \beta f(h)} \cdot \E[Q_x + Q_y].
\]
This shows that 
\[
\E[|L_x + L_y|] \quad \geq \quad 0.81 \cdot \frac{1}{11 \cdot \beta f(h)} \cdot \E[Q_x + Q_y] \quad \geq \quad \frac{1}{\beta^2 f(h)} \cdot \E[Q_x + Q_y],
\]
which finishes the entire proof.
\end{proofof}


\section{\pbthree via \pbtwo} 
\label{sec:envyMinim}

In this section we show how our results on \pbtwo can be used for the envy minimization problem. Our proofs go via a stronger notion of envy which we call the {\em ordinal envy}. 
To distinguish between the different notions of envy, we call the envy in terms of value (as in the definition of \pbthree problem)  the {\em cardinal envy}.

\subsection{Ordinal Envy}
\label{subsec:OrdinalEnvy}
Consider $n$ indivisible items and a single player with valuation $\valuation$. For simplicity, we assume that the player has different valuation for different items, i.e. $v_j \neq v_{j'}$ for any $j \neq j'$.
In the following, we give three \emph{equivalent}  definitions for ordinal envy of a player on being allocated a subset $S \subseteq [n]$ of items. We will prove the equivalence in Section~\ref{subsec:EquivProof}. 
We start with our first definition.

\begin{defn}[Ordinal Envy, First Definition]
\label{defn:OrdinalEnvy}
Suppose we are given $n$ indivisible items and a player's valuation $\valuation = (v_1,\cdots,v_n)$ of the items. We first relabel the items such that  $ v_1 > v_2 > \cdots > v_n$ and then define   \emph{ordinal envy} of the player on getting subset $S \subseteq [n]$ as
\[		
\envy_O(\valuation,S) ~\overset{\Def}{=}~ \max_{t \geq 0} \Big\{ |\overline{S} \cap [t]| - |S \cap [t]| \Big\}.
\]
\end{defn}
Notice that the ordinal envy  in Definition~\ref{defn:OrdinalEnvy} is always non-negative as we can take $[t]= \emptyset$. This definition considers every prefix of items after  sorting them in decreasing values. It  will therefore have direct connections to  interval discrepancy (Lemma~\ref{lem:DiscUpperBoundsOEnvy}).

Our second definition of ordinal envy is in terms of a {\em cancellation procedure} where items allocated to the player are used to cancel those that are not allocated to him and have lower prices. 
The ordinal envy is then defined to be the number of items that are not allocated to the player and not cancelled during this procedure.

\begin{defn}[Ordinal Envy, Second Definition]
\label{defn:OrdinalEnvyEquiv}
Suppose we are given $n$ indivisible items and a player's valuation $\valuation = (v_1,\cdots,v_n)$ of the items. 
Consider the following \emph{cancellation procedure}:
\begin{enumerate}[topsep=0mm,itemsep=0mm]
\item Initialize $T \leftarrow [n] \backslash S$;
\item For $i \in S$ in order of decreasing $v_i$:
	\item \qquad Set $j \leftarrow \arg\max_{j \in T} \{v_j|v_j < v_i\}$;
	\item \qquad Update $T \leftarrow T \backslash \{j\}$;
\item Return $T$.
\end{enumerate}
The ordinal envy of the player  on getting  subset $S \subseteq [n]$ is defined as
$ \envy_O(\valuation,S) \overset{\Def}{=} |T|.$
\end{defn}

For the third definition of ordinal envy, we need some notation. 
Given an ordering $\pi$ and a valuation $\valuation$ of  $n$ items, let $\valuation \sim \pi$  denote that $\valuation$ is consistent with the ordering $\pi$. Moreover,  for any valuation $\valuation$, let $\pi_{\valuation}$ denote its induced ordering. We define ordinal envy to be the worst possible cardinal envy that is consistent with $\pi_{\valuation}$. 

\begin{defn}[Ordinal Envy, Third Definition]
\label{defn:OEnvyEqualsWorstVEnvy}
Suppose we are given $n$ indivisible items and a player's valuation $\valuation = (v_1,\cdots,v_n)$ of the items. 
Let $\pi_{\valuation}$ be the ordering induced by the valuation $\valuation$, i.e. $\valuation \sim \pi_{\valuation}$.
The ordinal envy of the player  on getting  subset $S \subseteq [n]$ is defined as the worst possible cardinal envy that is consistent with $\pi_{\valuation}$, i.e.
\[
\envy_O(\valuation, S) ~\overset{\Def}{=} \sup_{\valuation' \sim \pi_{\valuation}, \valuation' \in [0,1]^n} \envy_C(\valuation', S).
\]
\end{defn}

While the above three definitions of ordinal envy appear very different, the following theorem states that they are  equivalent. We defer the proof of this theorem to Section~\ref{subsec:EquivProof}, 

\begin{theorem}
\label{thm:OrdEnvyEquiv}
The three definitions of ordinal envy in Definition~\ref{defn:OrdinalEnvy},~\ref{defn:OrdinalEnvyEquiv}, and~\ref{defn:OEnvyEqualsWorstVEnvy} are equivalent.
\end{theorem}

We first give some simple corollaries that will be useful in later proofs.
The first corollary states that the ordinal envy only depends on the ordering of the valuations. 
This follows directly from Definition~\ref{defn:OEnvyEqualsWorstVEnvy}.
 
\begin{cor}
\label{cor:OEnvyDependsonOrder}
For any ordering $\pi$, any valuations $\valuation \sim \pi$ and $\valuation' \sim \pi$ that are consistent with $\pi$, and any subset $S \subseteq [n]$ of items, we have
\[
\envy_O(\valuation, S) ~=~ \envy_O(\valuation',S).
\]
\end{cor}

The second corollary states that for any valuation $\valuation \in [0,1]^n$, the ordinal envy always upper bounds the cardinal envy. This is again an immediate consequence of Defintion~\ref{defn:OEnvyEqualsWorstVEnvy}. 

\begin{cor}
\label{cor:OEnvyUpperBoundsVEnvy}
For any subset of items $S \subseteq [n]$ and for any valuation $\valuation \in [0,1]^n$  where $v_j \neq v_{j'}$ for all  $j \neq j'$, the  ordinal envy upper bounds the cardinal envy, i.e.,
\[
\envy_O(\valuation,S) \quad \geq \quad \envy_C(\valuation,S)   \quad \overset{\Def}{=} \quad \max\{ v(\overline{S}) - v(S), 0\}.
\]
\end{cor}

\subsection{Proving the Equivalence of Ordinal Envy Definitions}
\label{subsec:EquivProof}

In this section, we prove Theorem~\ref{thm:OrdEnvyEquiv} which shows the equivalence of the three definitions of ordinal envy.

\begin{proofof}{Theorem~\ref{thm:OrdEnvyEquiv}}
We start by proving the equivalence of Definition~\ref{defn:OrdinalEnvy} and Definition~\ref{defn:OrdinalEnvyEquiv}.
\begin{lemma} 
\label{lem:OrdinalEnvyEquiv}
Definition~\ref{defn:OrdinalEnvy} is equivalent to Definition~\ref{defn:OrdinalEnvyEquiv}.
\end{lemma}

\begin{proof}
We assume the items are relabelled such that $ v_1 > v_2 > \cdots > v_n$. To prove the equivalence, we consider the outcome $T \subseteq [n]$ of the cancellation procedure in Definition~\ref{defn:OrdinalEnvyEquiv}. We prove  that 
\begin{align}\label{eq:eqOrdDefns}
\max_{t \geq 0} \Big\{ |\overline{S} \cap [t]| - |S \cap [t]| \Big\} ~= ~|T|.
\end{align}
Denote $P$ the set of all pairs of cancellation, i.e., $(i,j) \in P$ if and only if $i \in S$, $j \in \overline{S}$ and $i$ cancels $j$ in the cancellation procedure.
We call an item $i \in S$ \emph{matched} if there exists $j \in \overline{S}$ s.t. $(i,j) \in P$.
If $T = \emptyset$, then for any $t \in [n]$ we have that each $j \in [t] \cap \overline{S}$ gets canceled by some matched item $i \in S$ with $v_i > v_j$. Since $v_i > v_j$, we know that item $i \in [t]$. Thus, we have 
\[
|\overline{S} \cap [t]| - |S \cap [t]| \quad \leq \quad 0 \quad=\quad |T|.
\]
In particular, picking $t = 0$ gives $|\overline{S} \cap [0]| - |S \cap [0]| = 0 = |T|$.
Therefore, if $T = \emptyset$ then~\eqref{eq:eqOrdDefns} is true.

From now on we assume without loss of generality that $T \neq \emptyset$.
Consider $j^* = \arg \min_{j \in T} \{v_j\}$.
Notice that there is no pair  $(i,j) \in P$ with $v_i > v_{j^*}$ but $v_j < v_{j^*}$. This is because otherwise $i \in S$ would have canceled $j^* \in \overline{S}$ instead of $j \in \overline{S}$.
This implies that any matched item $i \in S$ with value $v_i > v_{j^*}$ cancels some other item $j \in \overline{S}$ with $v_j > v_{j^*}$. 
Also notice that any matched item $i \in S$ with value $v_i < v_{j^*}$ cancels an item $j \in \overline{S}$ with $v_j < v_i < v_j^*$.
Therefore, cancellation happens either completely inside the subset of items $[j^*]$ or completely outside $[j^*]$.
By definition of $j^*$, each $j \in T$ satisfies that $v_j \in [j^*]$. 
We therefore have
\begin{align}
\label{eqn:envyOMaxInterval}
|\overline{S} \cap [j^*]| - |S \cap [j^*]|  \quad = \quad  |T|.
\end{align}
Now notice that for each $t \in [n]$ and each $j \in \overline{S} \cap [t]$ that gets canceled by some  $i \in S$, we must have $v_i > v_j$. This implies that $i \in S \cap [t]$. Thus for any $t \in [n]$, we have 
\[
\big|\overline{S} \cap [t] \big| - \big|S \cap [t]\big| ~\leq~ |T|.
\]
Since by~\eqref{eqn:envyOMaxInterval} we can achieve equality in this equation for $t = j^*$, this again implies~\eqref{eq:eqOrdDefns} is true, finishing the proof of Lemma~\ref{lem:OrdinalEnvyEquiv}. 
\end{proof}

Now we prove that Definition~\ref{defn:OrdinalEnvyEquiv} is equivalent to Definition~\ref{defn:OEnvyEqualsWorstVEnvy}. This will complete the proof of the theorem.

\begin{lemma}
\label{lem:Def2EquivDef3}
Definition~\ref{defn:OrdinalEnvyEquiv} is equivalent to Definition~\ref{defn:OEnvyEqualsWorstVEnvy}.
\end{lemma}

\begin{proof}
We assume that the items are relabelled such that  $ v_1 > v_2 > \cdots > v_n$. 
Throughout the proof of the lemma, we use $\envy_O(\valuation,S)$ to denote the ordinal envy as in Definition~\ref{defn:OrdinalEnvyEquiv}.
To prove Lemma~\ref{lem:Def2EquivDef3}, we only need to prove 
\begin{align}
\label{eqn:GetDefn3}
\envy_O(\valuation, S) ~= \sup_{\valuation' \sim \pi_{\valuation}, \valuation' \in [0,1]^n} \envy_C(\valuation', S).
\end{align}

We start by showing that for any $\valuation' \in [0,1]^n$ such that $\valuation' \sim \pi_{\valuation}$, we have
\begin{align}
\label{eqn:OEnvyUpperBoundsVEnvy}
\envy_O(\valuation,S) \quad \geq \quad \envy_C(\valuation',S)   \quad \overset{\Def}{=} \quad \max\{ v'(\overline{S}) - v'(S), 0\}.
\end{align}

Notice we only need to prove $v'(\overline{S}) - v'(S) \leq \envy_O(\valuation',S)$ since $\envy_O(\valuation, S) \geq 0$ and that $\envy_O(\valuation, S) = \envy_O(\valuation', S)$ by Definition~\ref{defn:OrdinalEnvyEquiv}.
Consider the outcome $T$ of the cancellation procedure in Definition~\ref{defn:OrdinalEnvyEquiv} on valuation $\valuation'$. Denote $P$ the set of all pairs of cancellation, i.e. $(i,j) \in P$ if and only if $i \in S$, $j \in \overline{S}$ and $i$ cancels $j$ during the cancellation process. Denote
\[
S' ~\overset{\Def}{=}~ \big\{i \in S \mid \nexists j \in \overline{S} \text{ s.t. } (i,j) \in P\big\}
\] 
the set of items in $S$ that are not used to cancel any item in $\overline{S}$. Also recall that for each pair $(i,j) \in P$, we have $v'_i > v'_j$. 
Therefore, we have
\begin{align}
\label{eqn:CardinalEnvyDecomposition}
v'(\overline{S}) - v'(S) ~=~ \sum_{j \in T} v'_j + \sum_{(i,j) \in P} (v'_j - v'_i) - \sum_{i \in S'} v'_i.
\end{align}
Since each $v'_j \in [0,1]$, the first term in~(\ref{eqn:CardinalEnvyDecomposition}) is bounded as $\sum_{j \in T} v'_j \leq |T|$.
The second term in~(\ref{eqn:CardinalEnvyDecomposition}) is bounded as $\sum_{(i,j) \in P} (v'_j - v'_i) \leq 0$. Finally, the third term in~(\ref{eqn:CardinalEnvyDecomposition}) is bounded as $\sum_{i \in S'} v'_i \geq 0$.
We therefore have
\[
v'(\overline{S}) - v'(S) \quad \leq \quad |T| \quad = \quad \envy_O(\valuation', S),
\]
which establishes~(\ref{eqn:OEnvyUpperBoundsVEnvy}).

Notice that~(\ref{eqn:OEnvyUpperBoundsVEnvy}) immediately implies that the ordinal envy as in Definition~\ref{defn:OrdinalEnvyEquiv} satisfies that
\[
\envy_O(\valuation, S) ~\geq~ \sup_{\valuation' \sim \pi_{\valuation}, \valuation' \in [0,1]^n} \envy_C(\valuation', S).
\] 
In the following we show that for any $\delta > 0$, we have
\begin{align}
\label{eqn:WorstVEnvyApproachesOEnvy}
\sup_{\valuation' \sim \pi_{\valuation}, \valuation' \in [0,1]^n} \envy_C(\valuation', S) ~\geq ~ \envy_O(\valuation, S) - \delta.
\end{align}
Once~(\ref{eqn:WorstVEnvyApproachesOEnvy}) is established, (\ref{eqn:GetDefn3}) follows immediately by taking $\delta \rightarrow 0$ and this will finish the proof of the lemma.

We take $\epsilon = \delta/n^2$. By Definition~\ref{defn:OrdinalEnvy} (which is equivalent to Definition~\ref{defn:OrdinalEnvyEquiv} by Lemma~\ref{lem:OrdinalEnvyEquiv}), there exists $t^*$ such that $\envy_O(\valuation, S) = \big|\overline{S} \cap [t^*]\big| - \big|S \cap [t^*] \big|$. We consider the following valuation $\valuation' \in [0,1]^n$ with\footnote{To satisfy $v'_j \neq v'_{j'}$ for all $j \neq j'$, we can instead take $v'_i$ for each $i \notin [t^*]$ to be some different value arbitrarily close to 0 which is consistent with $\pi_{\valuation}$.}
\begin{align}
\label{eqn:BadValuation}
v'_i ~= ~\begin{cases}
1 - i\epsilon \qquad &\text{if $i \in [t^*]$}\\
0 &\text{otherwise}.
\end{cases}
\end{align}
Notice that $\valuation' \sim \pi_{\valuation}$.
Since each item $i \notin [t^*]$ has 0 value in $\valuation'$, we have
\begin{align*}
\envy_C(\valuation', S) 
&= \sum_{i \in \overline{S} \cap [t^*]} v'_i - \sum_{i \in S \cap [t^*]} v'_i \\
&\overset{(\ref{eqn:BadValuation})}{\geq} \sum_{i \in \overline{S} \cap [t^*]} (1- n\epsilon) - \sum_{i \in S \cap [t^*]} 1 \\
&\geq \big|\overline{S} \cap [t^*]\big| - \big|S \cap [t^*] \big| - \epsilon n^2 \quad = \quad \envy_O(\valuation, S) - \delta.
\end{align*}
This establishes~(\ref{eqn:WorstVEnvyApproachesOEnvy}) and finishes the proof of the lemma.
\end{proof}
The last two lemmas imply Theorem~\ref{thm:OrdEnvyEquiv}.
\end{proofof}

\subsection{Discrepancy Upper Bounds Ordinal Envy}
Consider $n$ indivisible items and  a player with valuation $\valuation = (v_1,\cdots, v_n) \in [0,1]^n$ where $v_j \neq v_{j'}$ for  $j \neq j'$. Suppose the player is allocated a subset $S \subseteq [n]$ of  items. If for each item $i \in [n]$ we think of it as an arrival at $v_i \in [0,1]$ and color each element $i \in S$ with $\col(i) = -1$ and each element $i \in \overline{S}$ with $\col(i) = +1$, then we arrive at a coloring $\col$ of an \pbintdisc instance. We denote the discrepancy of the set of all intervals under the coloring $\col$ as $\disc_O(\valuation,S)$.
We show that $\disc_O(\valuation, S)$ always an upper bounds the player's ordinal envy $\envy_O(\valuation,S)$. 

\begin{lemma}
\label{lem:DiscUpperBoundsOEnvy}
For any valuation $\valuation$ of the player and any subset of items $S \subseteq [n]$, we have
\[
\disc_O(\valuation,S)  ~\geq~  \envy_O(\valuation,S).
\]
\end{lemma}
\begin{proof}
We assume that the items are relabelled such that  $ v_1 > v_2 > \cdots > v_n$. We first observe that 
\[
\disc_O(\valuation, S) ~\geq~ \max_{t \geq 0} \Big\{\left| |\overline{S} \cap [t]| - |S \cap [t]| \right| \Big\}.
\]
This is because for any $t \in [n]$, the interval $I = [v_t,1]$ is such that $[t] = \{i \in [n]: v_i \in I\}$. Now by Definition~\ref{defn:OrdinalEnvy}, we have
\[
\envy_O(\valuation,S) \quad = \quad  \max_{t \geq 0} \{|\overline{S} \cap [t]| - |S \cap [t]|\} \quad \leq \quad \max_{t \geq 0} \Big\{\Big| |\overline{S} \cap [t]| - |S \cap [t]| \Big| \Big\}.
\]
This implies that
\[
\envy_O(\valuation,S) \quad \leq \quad \max_{t \geq 0} \Big\{\Big| |\overline{S} \cap [t]| - |S \cap [t]| \Big|\Big\} \quad \leq \quad \disc_O(\valuation, S),
\]
which finishes the proof of the lemma.
\end{proof}

\subsection{Reducing \pbthree to  \pbtwo}
We consider the \pbthree problem defined in Section~\ref{subsec:OnlineEnvyMin}. Recall,  we are given two probability distributions $D_i$ over $[0,1]$ for $i\in \{1,2\}$ and the goal is to allocate $n$ items to  minimize envy where item $j$'s valuation  $v^i_j \sim D_i$. We restate our main result for this problem.

\thmpbthree*

\begin{proofof}{Theorem~\ref{thm:envyMin}}
We assume without loss of generality that the valuation distributions $D_1,D_2$ are continuous\footnote{If either valuation distribution $D_i$ has a point mass, we replace the point mass with a uniform distribution around the neighborhood with the same mass. This affects the cardinal envy by a tiny amount.}. Denote $F_i$ the cumulative density function (CDF) of distribution $D_i$ for $i=1,2$. Recall that by probability integral transformation, $F_i(X)$ has a uniform distribution over $[0,1]$ when $X \sim D_i$. Consider any outcome of the valuations $\valuation^1,\valuation^2$ and any subset $S \subseteq [n]$ of items that are allocated to the first player. Notice that w.p. $1$, we have: (1) $v^i_j \neq v^i_{j'}$ for any $i \in \{ 1,2\}$ and $j \neq j'$, and (2) $v^i_j$ is in the interior of the support of $D_i$ for any $i  \in \{ 1,2\}$ and $j \in [n]$. Now by Corollary~\ref{cor:OEnvyUpperBoundsVEnvy}, we have 
\[
\envy_O(\valuation^1, S) ~\geq~ \envy_C(\valuation^1, S) \qquad \text{and} \qquad \envy_O(\valuation^2, \overline{S}) ~\geq~ \envy_C(\valuation^2, \overline{S}).
\]

We consider valuations $\widetilde{\valuation}^i $ where $\widetilde{v}^i_j = F_i(v^i_j)$ for $i  \in \{ 1,2\}$ and $j \in [n]$. Since the function $F_i:[0,1] \rightarrow [0,1]$ is monotone and each $v^i_j$ is in the interior of the support of $D_i$, we have that  $v^i_j > v^i_{j'}$ implies that $\widetilde{v}^i_j > \widetilde{v}^i_{j'}$ for $i  \in \{ 1,2\}$ and any $j,j' \in [n]$. Therefore, we have from Corollary~\ref{cor:OEnvyDependsonOrder} that 
\[
\envy_O(\widetilde{\valuation}^1, S) ~=~ \envy_O(\valuation^1, S) \qquad \text{and} \qquad \envy_O(\widetilde{\valuation}^2, \overline{S}) ~=~ \envy_O(\valuation^2, \overline{S}).
\]
It immediately follows that 
\[
\envy_O(\widetilde{\valuation}^1, S) ~\geq~ \envy_C(\valuation^1, S) \qquad \text{and} \qquad \envy_O(\widetilde{\valuation}^2, \overline{S}) ~\geq~ \envy_C(\valuation^2, \overline{S}).
\]
Notice that $\widetilde{v}^i_j$ is an independent random variable with uniform distribution on $[0,1]$. If we think of each item $j \in [n]$ as an element that arrives at $(\widetilde{v}^1_j, \widetilde{v}^2_j) \in [0,1] \times [0,1]$, then each element arrives uniformly at random in the unit square $[0,1] \times [0,1]$. Our Theorem~\ref{thm:StripeProblem} for \pbtwo problem gives an algorithm that achieves discrepancy $\disc = O( n^{c/\log \log n} )$ w.h.p. for some universal constant $c$ over the set of stripes $[a,b]\times [0,1]$ and $[0,1]\times [a,b]$ for $0\leq a <b \leq 1$. Using Lemma~\ref{lem:DiscUpperBoundsOEnvy}, we have that 
\[
\disc ~ \geq ~ \envy_O(\widetilde{\valuation}^1, S) \qquad \text{and} \qquad \disc ~ \geq ~ \envy_O(\widetilde{\valuation}^2, \overline{S}),
\]
which implies that 
\[
\disc ~ \geq ~ \envy_C(\valuation^1, S) \qquad \text{and} \qquad \disc ~ \geq ~ \envy_C(\valuation^2, \overline{S}).
\]
It follows that the envy, as defined in~(\ref{eq:envyDefn}), is bounded by $O(n^{c/\log \log n})$ w.h.p. for some universal constant $c$.
This completes the proof of Theorem~\ref{thm:envyMin}.
\end{proofof}

\medskip
\noindent
{\bf Acknowledgments}.
We are grateful to Alex Psomas  and Ariel Procaccia 
for introducing us to the online envy minimization problem. Part of this research was done while the first and the last author were visiting Microsoft Research, Redmond.

\appendix


\section{Missing Proofs in Section~\ref{sec:OnlineTreeBal}}
\label{sec:MissProofSec2}
\begin{proofof}{Fact~\ref{claim:DangerousSetRatioBound}}
We assume without loss of generality that $x \geq \frac{\log n}{\lambda}$ (the case where $x \leq - \frac{\log n}{\lambda}$ is similar). Since $y \in \dang(x)$, we have $|y| \in \left[x - \frac{\log 10}{\lambda}, x + \frac{\log 10}{\lambda} \right]$. It follows that 
\begin{align*}
\frac{\cosh(\lambda x)}{\cosh(\lambda y)} 
\quad \leq \quad \frac{\cosh(\lambda x)}{\cosh(\lambda x-\log 10)} \quad = \quad \frac{(e^{\lambda x} + e^{-\lambda x})/2}{(e^{\lambda x-\log 10} + e^{-\lambda x + \log 10})/2} 
\quad \leq \quad \frac{e^{\lambda x} + 1}{\frac{1}{10} \cdot e^{\lambda x}} 
\quad \leq \quad 11,
\end{align*} 
where in the last inequality we use the fact that $e^{\lambda x} \geq n \gg 1$. Similarly we also have 
\begin{align*}
\frac{\cosh(\lambda y)}{\cosh(\lambda x)} 
\quad \leq \quad \frac{\cosh(\lambda x + \log 10)} {\cosh(\lambda x)}
\quad= \quad \frac{(e^{\lambda x + \log 10} + e^{-\lambda x - \log 10})/2}{(e^{\lambda x} + e^{-\lambda x})/2} 
\quad \leq \quad \frac{10 \cdot e^{\lambda x} + 1}{e^{\lambda x}} 
\quad \leq \quad 11.
\end{align*} 
This finishes the proof of the inequality
$\max \left\{ \frac{\cosh(\lambda x)}{\cosh(\lambda y)} ~,~ \frac{\cosh(\lambda y)}{\cosh(\lambda x)} \right\} ~ \leq ~ 11.$

The proof of the other inequality for $|\sinh(\cdot)|$ is similar.
\end{proofof}

\begin{proofof}{Fact~\ref{claim:DangerousSetNoCancel}}
We assume without loss of generality that $x \geq \frac{\log n}{\lambda}$ (the case where $x \leq - \frac{\log n}{\lambda}$ is similar). Since $y \notin \dang(x)$, we have the following three cases: (1)  $y \geq 0$, or (2) $y \in \big(-x + \frac{\log 10}{\lambda}, 0 \big)$, or (3) $y < -x - \frac{\log 10}{\lambda}$.

\noindent \textbf{Case 1: $y \geq 0$.} In this case $\sinh(\lambda x)$ and $\sinh(\lambda y)$ don't cancel each other and we have
\[
|\sinh(\lambda x) + \sinh(\lambda y)| \quad = \quad |\sinh(\lambda x)| + |\sinh(\lambda y)| \quad \geq \quad \max \big\{|\sinh(\lambda x)|~,~ |\sinh(\lambda y)|\big\}.
\]
So we are done in this case.

\noindent \textbf{Case 2: $y \in \left( -x + \frac{\log 10}{\lambda}, 0 \right)$.} In this case we have $|x| \geq |y|$ which implies that $|\sinh(\lambda x)| \geq |\sinh(\lambda y)|$.
Notice that 
\begin{align*}
\frac{|\sinh(\lambda x)|}{|\sinh(\lambda y)|} 
\quad\geq \quad \frac{|\sinh(\lambda x)|}{|\sinh(\lambda x - \log 10)|} \quad \geq \quad \frac{(e^{\lambda x} - e^{-\lambda x}) / 2}{(e^{\lambda x - \log 10} - e^{-\lambda x + \log 10}) / 2} \quad \geq \quad  \frac{e^{\lambda x} - 1}{\frac{1}{10} \cdot e^{\lambda x}} \quad \geq \quad 9,
\end{align*}
where in the last inequality we use the fact that $e^{\lambda x} \geq n \gg 1$. It follows that
\[
|\sinh(\lambda x) + \sinh(\lambda y)| ~=~ |\sinh(\lambda x)| - |\sinh(\lambda y)| ~\geq~ \frac{8}{9} \cdot |\sinh(\lambda x)| ~=~ \frac{8}{9} \cdot \max \big\{|\sinh(\lambda x)|, |\sinh(\lambda y)|\big\}.
\]
So we are done in this case.

\noindent \textbf{Case 3: $y < -x - \frac{\log 10}{\lambda}$.} In this case we have $|x| \leq |y|$ which implies that $|\sinh(\lambda x)| \leq |\sinh(\lambda y)|$. A similar argument as in the previous case gives $\frac{|\sinh(\lambda y)|}{|\sinh(\lambda x)|} \geq 9$. It follows that
\[
|\sinh(\lambda x) + \sinh(\lambda y)| = |\sinh(\lambda y)| - |\sinh(\lambda x)| \geq \frac{8}{9} \cdot |\sinh(\lambda y)| = \frac{8}{9} \cdot \max \big\{|\sinh(\lambda x)|, |\sinh(\lambda y)|\big\}.
\]
This finishes the proof of Fact~\ref{claim:DangerousSetNoCancel}.
\end{proofof}


\section{Tightness of the Separation Lemma} \label{sec:SepLemmaTightness}

 Consider  an $m$-ary tree with height $h = o(\log n / \log \log n)$ and $m = n^{1/h}$.
We choose $x = \frac{\log n + h \log 10}{\lambda}$ and $d_r = \frac{5 \log n}{\lambda}$. 
This ensures that $\cosh(\lambda x) \ll \cosh(\lambda d_r)$ so we can ignore $\sinh(\lambda x)$ and $\cosh(\lambda x)$ while losing only a tiny additive factor.

For each $i \in [h]$, define $d_{i,-}$ to be such that $\sinh(\lambda d_{i,-}) = -\left( 1 + \frac{1}{\log n} \right)^{i-1} \sinh(\lambda d_r)$ and $d_{i,+}$ to be such that $\sinh(\lambda d_{i,+}) = \left(1 + \frac{1}{\log n} \right)^{i-1} \cdot \frac{\sinh(\lambda d_r)}{\log n}$.
Notice that for any $i \in [h]$, we have $d_{i,-} < 0$ and $d_{i,+} > 0$.
Define $p_i = \frac{|d_{i,-}|}{|d_{i,+}| + |d_{i,-}|}$.

Now we are ready to define the imbalance of the nodes in the tree $\T$.
Denote the depth-$i$ node as $L_i(\T)$ for $i \in [h]$ and denote $L_i(\T,v) \subseteq L_i(\T)$ for $v \in L_{i-1}(\T)$ the set of children of $v$.
Each $s \in L_1(\T)$ has imbalance either $d_{1,-}$ or $d_{1,+}$.
Notice the fraction of $L_1(\T)$ with imbalance $d_{1,+}$ is roughly $p_1$.
For each $s \in L_1(\T)$ with $d_s = d_{1,-}$, we make the rest of the sub-tree $\T_s$ rooted at $s$ to be roughly empty.
Essentially whenever an $s \in L_1(\T)$ with $d_s = d_{1,-}$ is picked by the root-leaf path $\calP_t$, then the $L$ term cancels out and becomes almost 0.
For $s \in L_1(\T)$ with $d_s = d_{1,+}$, we assign imbalance either $d_{2,-}$ or $d_{2,+}$ to nodes in $L_2(\T, s)$.
The fraction of nodes in $L_2(\T,s)$ with imbalance $d_{2,+}$ is roughly $p_2$.
We continue this pattern until we reach depth $h$: whenever we enter a node with negative imbalance, we make it's sub-tree almost empty and if we continue to enter a depth-$i$ node $v$ with positive imbalance, we assign imbalance $d_{i+1,-}$ or $d_{i+1,+}$ to $L_{i+1}(\T,v)$ and we have that the fraction of children with imbalance $d_{i+1,+}$ is roughly $p_{i+1}$.

Now suppose we pick a uniformly random root-leaf path $\calP$ in such a tree $\T$. 
Notice that whenever $\calP$ enters a node with negative imbalance, the $L$ term will become close to 0. 
So the only case where the $L$ term is large is when every node in $\calP$ has positive imbalance but this happens with probability $\prod_{i \in [h]} p_i =  O(1) \cdot 2^{-h}$.
The $L$ term in this case is $\cosh(\lambda d_r) + \sum_{i \in [h]} \cosh(\lambda d_{i,+}) = \left( 1 + \frac{1}{\log n} \right)^h \cosh(\lambda d_r) = O(1) \cdot \cosh(\lambda d_r)$.
Therefore, we have $\E[|L|] = O(1) \cdot 2^{-h} \cdot \cosh(\lambda d_r) \leq  2^{-\Omega(h)} \cdot \E[Q] - hn^2$, where the last inequality follows from the fact that $\cosh(\lambda d_r) \gg hn^2$.

\bibliographystyle{alpha}
\bibliography{bib}

\begin{thebibliography}{AGMW15}

\bibitem[AGMW15]{AzizGMW-AI15}
Haris Aziz, Serge Gaspers, Simon Mackenzie, and Toby Walsh.
\newblock Fair assignment of indivisible objects under ordinal preferences.
\newblock {\em Artificial Intelligence}, 227:71--92, 2015.

\bibitem[AS16]{alonspencer}
Noga Alon and Joel~H Spencer.
\newblock {\em The probabilistic method}.
\newblock John Wiley \& Sons, 2016.

\bibitem[Ban98]{Banaszczyk-Journal98}
Wojciech Banaszczyk.
\newblock {Balancing vectors and Gaussian measures of n-dimensional convex
  bodies}.
\newblock {\em Random Structures \& Algorithms}, 12(4):351--360, 1998.

\bibitem[Ban10]{Bansal-FOCS10}
Nikhil Bansal.
\newblock {Constructive Algorithms for Discrepancy Minimization}.
\newblock In {\em 51th Annual {IEEE} Symposium on Foundations of Computer
  Science, {FOCS} 2010, October 23-26, 2010, Las Vegas, Nevada, {USA}}, pages
  3--10, 2010.

\bibitem[Ban19]{Bansal-Notes19}
Nikhil Bansal.
\newblock {Discrepancy and Combinatorial Optimization Lecture 1-IPCO summer
  school}.
\newblock 2019.

\bibitem[BDG16]{BansalDG16}
Nikhil Bansal, Daniel Dadush, and Shashwat Garg.
\newblock An algorithm for koml{\'{o}}s conjecture matching banaszczyk's bound.
\newblock In {\em {IEEE} 57th Annual Symposium on Foundations of Computer
  Science, {FOCS} 2016, 9-11 October 2016, Hyatt Regency, New Brunswick, New
  Jersey, {USA}}, pages 788--799, 2016.

\bibitem[BDGL18]{BansalDGL18}
Nikhil Bansal, Daniel Dadush, Shashwat Garg, and Shachar Lovett.
\newblock The gram-schmidt walk: a cure for the banaszczyk blues.
\newblock In {\em Proceedings of the 50th Annual {ACM} {SIGACT} Symposium on
  Theory of Computing, {STOC} 2018, Los Angeles, CA, USA, June 25-29, 2018},
  pages 587--597, 2018.

\bibitem[Bec81]{Beck-Combinatorica81}
J{\'o}zsef Beck.
\newblock {Balanced two-colorings of finite sets in the square I}.
\newblock {\em Combinatorica}, 1(4):327--335, 1981.

\bibitem[BF81]{BeckFiala-DAM81}
J{\'o}zsef Beck and Tibor Fiala.
\newblock {``Integer-making'' theorems}.
\newblock {\em Discrete Applied Mathematics}, 3(1):1--8, 1981.

\bibitem[BG17]{BansalG17}
Nikhil Bansal and Shashwat Garg.
\newblock Algorithmic discrepancy beyond partial coloring.
\newblock In {\em Proceedings of the 49th Annual {ACM} {SIGACT} Symposium on
  Theory of Computing, {STOC} 2017, Montreal, QC, Canada, June 19-23, 2017},
  pages 914--926, 2017.

\bibitem[BKPP18]{BenadeKPP-EC18}
Gerdus Benade, Aleksandr~M. Kazachkov, Ariel~D. Procaccia, and
  Christos{-}Alexandros Psomas.
\newblock {How to Make Envy Vanish Over Time}.
\newblock In {\em Proceedings of the 2018 {ACM} Conference on Economics and
  Computation, Ithaca, NY, USA, June 18-22, 2018}, pages 593--610, 2018.

\bibitem[BM19]{BansalMeka-SODA19}
Nikhil Bansal and Raghu Meka.
\newblock {On the discrepancy of random low degree set systems}.
\newblock In {\em Proceedings of the Thirtieth Annual {ACM-SIAM} Symposium on
  Discrete Algorithms, {SODA} 2019, San Diego, California, USA, January 6-9,
  2019}, pages 2557--2564, 2019.

\bibitem[BS13]{Bansal2013}
Nikhil Bansal and Joel Spencer.
\newblock Deterministic discrepancy minimization.
\newblock {\em Algorithmica}, 67(4):451--471, Dec 2013.

\bibitem[BS19]{BansalSpencer-arXiv19}
Nikhil Bansal and Joel~H. Spencer.
\newblock On-line balancing of random inputs.
\newblock {\em CoRR}, abs/1903.06898, 2019.

\bibitem[Bud11]{Budish-Journal11}
Eric Budish.
\newblock The combinatorial assignment problem: Approximate competitive
  equilibrium from equal incomes.
\newblock {\em Journal of Political Economy}, 119(6):1061--1103, 2011.

\bibitem[Cha01]{Chazelle-Book01}
Bernard Chazelle.
\newblock {\em {The discrepancy method: randomness and complexity}}.
\newblock Cambridge University Press, 2001.

\bibitem[EL19]{EzraL19}
Esther Ezra and Shachar Lovett.
\newblock On the beck-fiala conjecture for random set systems.
\newblock {\em Random Struct. Algorithms}, 54(4):665--675, 2019.

\bibitem[ES18]{EldanS18}
Ronen Eldan and Mohit Singh.
\newblock Efficient algorithms for discrepancy minimization in convex sets.
\newblock {\em Random Struct. Algorithms}, 53(2):289--307, 2018.

\bibitem[Fol67]{FoleyEssay67}
Duncan~K Foley.
\newblock Resource allocation and the public sector.
\newblock {\em Yale Econ Essays}, 7:45--98, 1967.

\bibitem[FS18]{Franks19}
Cole Franks and Michael Saks.
\newblock On the discrepancy of random matrices with many columns.
\newblock {\em arXiv}, 1807.04318, 2018.

\bibitem[HR17]{HobergR17}
Rebecca Hoberg and Thomas Rothvoss.
\newblock A logarithmic additive integrality gap for bin packing.
\newblock In {\em Proceedings of the Twenty-Eighth Annual {ACM-SIAM} Symposium
  on Discrete Algorithms, {SODA} 2017, Barcelona, Spain, Hotel Porta Fira,
  January 16-19}, pages 2616--2625, 2017.

\bibitem[HR19]{HobergRothvoss-SODA19}
Rebecca Hoberg and Thomas Rothvoss.
\newblock {A Fourier-Analytic Approach for the Discrepancy of Random Set
  Systems}.
\newblock In {\em Proceedings of the Thirtieth Annual {ACM-SIAM} Symposium on
  Discrete Algorithms, {SODA} 2019, San Diego, California, USA, January 6-9,
  2019}, pages 2547--2556, 2019.

\bibitem[LM15]{Lovett-Meka-SICOMP15}
Shachar Lovett and Raghu Meka.
\newblock {Constructive Discrepancy Minimization by Walking on the Edges}.
\newblock {\em {SIAM} J. Comput.}, 44(5):1573--1582, 2015.

\bibitem[LMMS04]{LiptonMMS-EC04}
Richard~J. Lipton, Evangelos Markakis, Elchanan Mossel, and Amin Saberi.
\newblock On approximately fair allocations of indivisible goods.
\newblock In {\em Proceedings 5th {ACM} Conference on Electronic Commerce
  (EC-2004), New York, NY, USA, May 17-20, 2004}, pages 125--131, 2004.

\bibitem[LRR17]{LevyRR17}
Avi Levy, Harishchandra Ramadas, and Thomas Rothvoss.
\newblock Deterministic discrepancy minimization via the multiplicative weight
  update method.
\newblock In {\em Integer Programming and Combinatorial Optimization - 19th
  International Conference, {IPCO} 2017, Waterloo, ON, Canada, June 26-28,
  2017, Proceedings}, pages 380--391, 2017.

\bibitem[Mat09]{Matousek-Book09}
Jiri Matousek.
\newblock {\em {Geometric discrepancy: An illustrated guide}}, volume~18.
\newblock Springer Science \& Business Media, 2009.

\bibitem[Nik14]{Nikolov-Thesis14}
Aleksandar Nikolov.
\newblock {\em New computational aspects of discrepancy theory}.
\newblock PhD thesis, Rutgers University-Graduate School-New Brunswick, 2014.

\bibitem[Nik17]{Nikolov-Mathematika19}
Aleksandar Nikolov.
\newblock Tighter bounds for the discrepancy of boxes and polytopes.
\newblock {\em CoRR}, abs/1701.05532, 2017.

\bibitem[NNN12]{NewmanNN-FOCS12}
Alantha Newman, Ofer Neiman, and Aleksandar Nikolov.
\newblock Beck's three permutations conjecture: {A} counterexample and some
  consequences.
\newblock In {\em 53rd Annual {IEEE} Symposium on Foundations of Computer
  Science, {FOCS} 2012, New Brunswick, NJ, USA, October 20-23, 2012}, pages
  253--262, 2012.

\bibitem[Rot14]{Rothvoss14}
Thomas Rothvo{\ss}.
\newblock Constructive discrepancy minimization for convex sets.
\newblock In {\em 55th {IEEE} Annual Symposium on Foundations of Computer
  Science, {FOCS} 2014, Philadelphia, PA, USA, October 18-21, 2014}, pages
  140--145, 2014.

\bibitem[Spe77]{Spencer77}
Joel Spencer.
\newblock Balancing games.
\newblock {\em Journal of Combinatorial Theory, Series B}, 23(1):68--74, 1977.

\bibitem[Spe85]{Spencer85}
Joel Spencer.
\newblock Six standard deviations suffice.
\newblock {\em Transactions of the American Mathematical Society}, pages
  289(2):679–706,, 1985.

\bibitem[Spe87]{Spencer-Book87}
Joel~H. Spencer.
\newblock {\em Ten lectures on the probabilistic method}, volume~52.
\newblock Society for Industrial and Applied Mathematics Philadelphia, 1987.

\bibitem[SST97]{SpencerST-SODA97}
Joel~H. Spencer, Aravind Srinivasan, and Prasad Tetali.
\newblock The discrepancy of permutation families.
\newblock In {\em SODA}, 1997.

\bibitem[TV85]{ThomsonVarian-Essay85}
William Thomson and Hal Varian.
\newblock Theories of justice based on symmetry.
\newblock {\em Social goals and social organizations: essays in memory of
  Elisha Pazner}, 126, 1985.

\end{thebibliography}

\end{document}